\documentclass[conference]{IEEEtran}
\usepackage{times}

\usepackage{amsmath}
\usepackage{amsthm}
\usepackage{amssymb}
\usepackage{amsfonts}
\usepackage{subcaption}
\usepackage{graphicx}
\usepackage{color}
\usepackage{url}
\usepackage[usenames,x11names]{xcolor}
\usepackage{xspace} % for inserting a space in TeX commands if needed
\usepackage{caption}
\usepackage[numbers]{natbib}
\usepackage{multicol}
\usepackage[normalem]{ulem}

\usepackage[linesnumbered,vlined,ruled]{algorithm2e}
\usepackage{multirow} % for cell tables spanning multiple rows
\usepackage{tikz,pgf}
\usetikzlibrary{arrows,automata,shapes,calc,backgrounds,spy,positioning,shadows}
\graphicspath{{figures/}}

\SetKw{Break}{break}

\newtheorem{theorem}{Theorem}[section]
\newtheorem{proposition}[theorem]{Proposition}
\newtheorem{corollary}[theorem]{Corollary}
\newtheorem{definition}{Definition}[section]

\newtheorem{remark}[theorem]{Remark}

\newtheorem{problem}{Problem}[section]

\newtheorem{assumption}{Assumption}

\newcommand{\CA}[1]{\mathcal{#1}}
\newcommand{\BF}[1]{\mathbf{#1}}
\newcommand{\BB}[1]{\mathbb{#1}}

\newcommand{\BS}[1]{\boldsymbol{#1}}

\newcommand{\notltl}{\neg}
\newcommand{\andltl}{\wedge}
\newcommand{\orltl}{\vee}

\newcommand{\Always}{\BF{G}}
\newcommand{\Event}{\BF{F}}

\newcommand{\Implies}{\Rightarrow}

\newcommand{\True}{\top}
\newcommand{\False}{\perp}

\newcommand{\trap}{\bowtie}

\newcommand{\AP}{AP}

\newcommand{\PA}{\mathcal{P}}
\newcommand{\BA}{\mathcal{B}}
\newcommand{\FA}{\mathcal{A}}
\newcommand{\TS}{\mathcal{T}}

\newcommand{\la}{\leftarrow}
\newcommand{\ra}{\rightarrow}
\newcommand{\ras}[1]{\stackrel{#1}{\rightarrow}}
\newcommand{\asgn}{\la}

\newcommand{\norm}[1]{\left\| {#1} \right\|}

\newcommand{\card}[1]{\left| {#1} \right|}
\newcommand{\spow}[1]{2^{#1}}
\newcommand{\trel}[1]{\left| {#1} \right|_{TR}}

\newcommand{\prefix}[1]{P\left({#1}\right)}
\begin{document}

% paper title
\title{Time Window Temporal Logic}

% You will get a Paper-ID when submitting a pdf file to the conference system
%\author{Author Names Omitted for Anonymous Review. Paper-ID [add your ID here]}

\author{\IEEEauthorblockN{Cristian-Ioan Vasile}
\IEEEauthorblockA{Division of Systems Engineering\\
Boston University\\
Brookline, Massachusetts 02446\\
Email: cvasile@bu.edu}
\and
\IEEEauthorblockN{Derya Aksaray}
\IEEEauthorblockA{Department of Mechanical Engineering\\
Boston University\\
Boston, Massachusetts 02215\\
Email: daksaray@bu.edu}
\and
\IEEEauthorblockN{Calin Belta}
\IEEEauthorblockA{Department of Mechanical Engineering\\
Boston University\\
Boston, Massachusetts 02215\\
Email: cbelta@bu.edu}}

% avoiding spaces at the end of the author lines is not a problem with
% conference papers because we don't use \thanks or \IEEEmembership

% for over three affiliations, or if they all won't fit within the width
% of the page, use this alternative format:
% 
%\author{\authorblockN{Michael Shell\authorrefmark{1},
%Homer Simpson\authorrefmark{2},
%James Kirk\authorrefmark{3}, 
%Montgomery Scott\authorrefmark{3} and
%Eldon Tyrell\authorrefmark{4}}
%\authorblockA{\authorrefmark{1}School of Electrical and Computer Engineering\\
%Georgia Institute of Technology,
%Atlanta, Georgia 30332--0250\\ Email: mshell@ece.gatech.edu}
%\authorblockA{\authorrefmark{2}Twentieth Century Fox, Springfield, USA\\
%Email: homer@thesimpsons.com}
%\authorblockA{\authorrefmark{3}Starfleet Academy, San Francisco, California 96678-2391\\
%Telephone: (800) 555--1212, Fax: (888) 555--1212}
%\authorblockA{\authorrefmark{4}Tyrell Inc., 123 Replicant Street, Los Angeles, California 90210--4321}}

\maketitle

\begin{abstract}
This paper introduces \emph{time window temporal logic} (TWTL), a rich expressivity language for describing various time bounded specifications. In particular, the syntax and semantics of TWTL enable the compact representation of serial tasks, which are typically seen in robotics and control applications. This paper also discusses the relaxation of TWTL formulae with respect to deadlines of tasks.
Efficient automata-based frameworks to solve synthesis, verification and learning problems are also presented.
The key ingredient to the presented solution is an algorithm to translate a TWTL formula to an
annotated finite state automaton that encodes all possible temporal relaxations of the specification.
Case studies illustrating the expressivity of the logic and the proposed algorithms are included.  
\end{abstract}

\IEEEpeerreviewmaketitle

\section{Introduction}
Temporal logic provides a mathematical formalism to reason about (concurrent) events in terms of time. Due to its rich expressivity, it has been widely used as a specification language to describe properties related to correctness, termination, mutual exclusion, reachability, or liveness~\cite{manna1981}.  
Recently, there has been a great interest in using temporal logic formulae in the analysis and control of dynamical systems for robotic applications. For example, linear temporal logic (LTL)~\cite{Baier08} has been extensively used in motion planning and control of robotic systems, e.g.,~\cite{ulusoy-ijrr2013,Karaman.Frazzoli:CDC08,aksaray2015,topcu2010,igor-icra2013,Belta-TRO05,Murray2009,KB-TAC08-LTLCon,fainekos2009,kress2009,LeZhVaOiScBe-ISER-2014}. 

In robotics applications, the tasks may involve some time constraints (e.g.,~\cite{solomon1987,pavone2009}). For example,
\begin{itemize}
\item[-] every visit to $A$ needs to be immediately followed by servicing $B$ within $5$ time units;

\item[-] two consecutive visits to $A$ need to be at least $10$ time units apart;

\item[-] visiting $A$ and servicing $B$ need to be completed before the time reaches $15$.
\end{itemize}

Such tasks cannot be described by LTL formulae since LTL cannot deal with temporal properties with explicit time constraints. Therefore, bounded temporal logics are used to capture the time constraints over the tasks. Some examples are bounded linear temporal logic (BLTL) \cite{Tkachev13,Clarke09}, metric temporal logic (MTL)~\cite{koymans1990}, or signal temporal logic (STL)~\cite{maler2004}.

%syntactically co-safe linear temporal logic (scLTL) \cite{kupferman2001}

In this paper, we propose a new specification language called {\it time window temporal logic} (TWTL). The semantics of TWTL is rich enough to express a wide variety of time-bounded specifications, e.g., ``Service $A$ for $3$ time units within the time interval $[0, 5]$ and after that service $B$ for $2$ time units within $[4, 9]$. If $C$ is serviced for $2$ time units within $9$ time units, then $D$ should be serviced for $3$ time units within the same time interval (i.e., within $9$ time units). For instance, some multi-robot persistent surveillance specifications are expressed as TWTL formulae in~\cite{vasile2014} and~\cite{aksaray2016}. Moreover, we define the notion called {\it temporal relaxation} of a TWTL formula, which is a quantity computed over the time intervals of a given TWTL formula. In this respect, if the temporal relaxation is 
%\color{blue}
\begin{itemize}
\item[-] {\it negative}, then the tasks expressed in the TWTL formula should be completed before their designated time deadlines, thus satisfying the relaxed formula implies the satisfaction of temporally more strict TWTL formula; 
\item[-] {\it zero}, then the relaxed formula is exactly same as the original TWTL formula; 
\item[-] {\it positive}, then some tasks expressed in the TWTL formula are allowed to be completed after their designated time deadlines, thus satisfying the relaxed formula implies the violation of the original TWTL formula (or the satisfaction of temporally less strict formula).
\end{itemize}
%\color{black}
We also present an automata-based framework for minimizing the temporal relaxation of a given TWTL formula in problems related to verification, synthesis, and learning.
%\textcolor{blue}{
In the theoretical computer science literature, finite languages and the complexity of construction their corresponding
automata have been extensively studied~\cite{Maia13,Han07,Campeanu01,Gao11,Daciuk2003}.
The algorithms proposed in this paper are specialized to handle TWTL formulae and produce the annotated automata,
which is used to solve synthesis, verification and learning problems efficiently.
%}
% that minimize temporal relaxation of TWTL formulae.
%in which temporal relaxation is optimized for the satisfaction of minimally relaxed TWTL formulae. Specifically, we demonstrate the potential of the proposed framework in various problems related to verification, synthesis, and parameter learning.

The proposed language TWTL has several advantages over the existing temporal logics. First, a desired specification can be represented in a more compact and
%\textcolor{blue}{
comprehensible
%}
%easy-to-use perceptible
way in TWTL than BLTL, MTL, or STL. For example, any deadlines expressed in a TWTL formula indicates the exact time bounds as opposed to an STL formula where the time bounds can be shifted.
%\textcolor{blue}{
Consider a specification as ``stay at $A$ for $4$ time steps within the time window $[0,10]$", which can be expressed in TWTL as $[H^4 A]^{[0,10]}$. The same specification can be expressed in STL as $F_{[0,10-4]} G_{[0,4]} A$ where the outermost time window needs to be modified with respect to the inner time window. Furthermore, compared to BLTL and MTL, the existence of explicit concatenation operator results in a compact representation for serial tasks that are prevalent in robotics and control applications. Under some mild assumptions, we provide a very efficient (linear-time) algorithm to handle concatenation of tasks. This is in contrast to the general result from computer science that concatenation of languages, even finite ones~\cite{Maia13}, is exponential in the worst case.
%}
Second, the notion of temporal relaxation enables a generic framework to construct the automaton of all possible relaxations of a TWTL formula. In literature, there are some studies investigating the control synthesis problems for minimal violations of LTL fragments \cite{Reyes13,Tumova-HSCC13,Tumova-IROS14,LiPrJoMurray-ICRA13,Guo15}. However, the special automaton proposed in this paper is a compact representation of all possible relaxations, which can be used in a variety of problems related to synthesis, verification, or learning to achieve minimal relaxations. Third, for a given TWTL formula, the complexity of constructing automata is independent of the corresponding time bounds.
To achieve this property, we exploit the structure of finite languages encoded by TWTL formulae. 
%Fourth, particularly for robotics and control applications, TWTL includes an explicit operator
%to define serial tasks, which is very prevalent in these application domains.

The main contributions of this paper are: 1) introducing a new specification language called TWTL, 2) defining {\it temporal relaxation} of a TWTL formula, 3) presenting a set of provably-correct algorithms to construct the automaton of a given TWTL formula (both for the relaxed and unrelaxed cases), 4) formulating a generic problem in terms of temporal relaxation of a TWTL formula, which can also be specialized into various problems such as verification, synthesis, or learning, and 5) developing a Python package to solve the three specialized problems.

\section{Preliminaries}
%Notation and notions:
%\begin{itemize}
%  \item set, cardinality, power set, Kleen/$\omega$-closure
%  \item words, length of words, language, sub-word
%  \item prefix language, quotient language induced by prefix equivalence, unambiguous language
%  \item DFA, language of a DFA
%\end{itemize}

In this section, we introduce the notation and briefly review
the main concepts from formal languages, automata
theory, and formal verification.
For a detailed exposition of these topics, the reader is refereed to \cite{Baier08,Hopcroft2006}
and the references therein.

Given $\BF{x}, \BF{x}' \in \BB{R}^n$, $n \geq 2$,
the relationship $\BF{x} \sim \BF{x}'$, where $\sim\in\{<, \leq , >, \geq\}$,
is true if it holds pairwise for all components.
$\BF{x} \sim a$ denotes $\BF{x} \sim a \BF{1}_n$,
where $a \in \BB{R}$ and $\BF{1}_n$ is the n-dimensional vector of all ones.
The extended set of real numbers is denoted by
$\overline{\BB{R}} = \BB{R} \cup \{\pm \infty\}$
% $q^{[d]}$ denotes $d$ repetitions of $q$.

Let $\Sigma$ be a finite set.
We denote the cardinality and the power set of $\Sigma$
by $\card{\Sigma}$ and $\spow{\Sigma}$, respectively.
A {\em word} over $\Sigma$ is a finite or infinite sequence of
elements from $\Sigma$. In this context, $\Sigma$ is also
called an {\em alphabet}. The length of a word $w$ is denoted by
$\card{w}$ (e.g., $\card{w}=\infty$ if $w$ is an infinite word).
Let $k$, $i\leq j$ be non-negative integers.
The $k$-th element of $w$ is denoted by $w_k$, and
the sub-word $w_i,\ldots, w_j$ is denoted by $w_{i, j}$.
A set of words over an alphabet $\Sigma$ is called
a {\em language} over $\Sigma$.
The languages of all finite and infinite words over $\Sigma$
are denoted by $\Sigma^*$ and $\Sigma^\omega$, respectively.

\begin{definition}[Prefix language]
\label{def:prefix-lang}
Let $\CA{L}_1$ and $\CA{L}_2$ be two languages.
We say that $\CA{L}_1$ is a prefix language of $\CA{L}_2$
if and only if every word in $\CA{L}_1$ is a prefix of
some word in $\CA{L}_2$, i.e.,
for each word $w \in \CA{L}_1$ there exists
$w^{\prime} \in \CA{L}_2$ such that $w = w^{\prime}_{0, i}$,
where $0 \leq i < \card{w^{\prime}}$,
The maximal prefix language of a language $\CA{L}$ is
denoted by
$\prefix{\CA{L}} = \{w_{0, i}\ |\ w\in \CA{L}, i\in \{0,\ldots,\card{w}-1\}\}$.
\end{definition}

%\begin{definition}[Prefix-quotient]
%Let $\CA{L}_1$ and $\CA{L}_2$ be two languages such that
%$\CA{L}_1$ is a prefix language of $\CA{L}_2$.
%The quotient language of $\CA{L}_2$ with respect to $\CA{L}_1$,
%denoted by $\CA{L}_2/\CA{L}_1$, is the set of words in $\CA{L}_2$
%which have a prefix in $\CA{L}_1$,
%i.e. $\CA{L}_2/\CA{L}_1 = \left\{w^{\prime} \in \CA{L}_2\, |\, \exists w \in \CA{L}_1, i \in \{0, \ldots \card{w^{\prime}}\} \text{ s.t. } w^{\prime}_{0, i} = w \right\}$.
%\end{definition}

\begin{definition}[Unambiguous language]
\label{def:unambiguous-lang}
A language $\CA{L}$ is called unambiguous language
if no proper subset $L$ of $\CA{L}$ is a prefix language
of $\CA{L}\setminus L$.
\end{definition}

The above definition immediately implies that a word
in an unambiguous language can not be the prefix of
another word. Moreover, it is easy to show that the
converse is also true.

\begin{definition}[Language concatenation]
\label{def:concat-lang}
Let $\CA{L}_1$ be a language over finite words, and
let $\CA{L}_2$ be a language over finite or infinite words.
The concatenation language $\CA{L}_1 \cdot \CA{L}_2$
is defined as the set of all words $ww^{\prime}$, where
$w\in \CA{L}_1$ and $w^{\prime} \in \CA{L}_2$.
\end{definition}

\begin{definition}[Deterministic Finite State Automaton]
\label{def:dfa}
A deterministic finite state automaton (DFA) is a tuple
$\FA = (S_\FA, s_0, \Sigma, \delta_\FA, F_\FA)$, where:
\begin{itemize}
    \item $S_\FA$ is a finite set of states;
    \item $s_0 \in S_\FA$ is the initial state;
    \item $\Sigma$ is the input alphabet;
    \item $\delta_\FA : S_\FA \times \Sigma \ra S_\FA$ is the transition function;
    \item $F_\FA \subseteq S_\FA$ is the set of accepting states.
\end{itemize}
\end{definition}

A transition $s' = \delta_\FA(s, \sigma)$ is also denoted by $s \ras{\sigma}_\FA s'$.
A trajectory of the DFA $\BF{s} = s_0 s_1 \ldots s_{n+1}$ is generated by
a finite sequence of symbols $\BS{\sigma} = \sigma_0 \sigma_1 \ldots \sigma_n$
if $s_0 \in S_\FA$ is the initial state of $\FA$ and
$s_k \ras{\sigma_k}_{\FA}  s_{k+1}$ for all $k \geq 0$.
The trajectory generated by $\BS{\sigma}$ is also denoted by
$s_0 \ras{\BS{\sigma}}_{\FA} s_{n+1}$.
A finite input word $\BS{\sigma}$ over $\Sigma$ is said to be accepted
by a finite state automaton $\FA$ if the trajectory of $\FA$ generated
by $\BS{\sigma}$ ends in a state belonging to the set of accepting states, i.e., $F_\FA$ .
A DFA is called {\em blocking} if the $\delta_\FA(s, \sigma)$ is a partial function,
i.e., the value of the function is not defined for all values in the domain.
A blocking automaton rejects words $\BS{\sigma}$ if
there exists $k \geq 0$ such that $s_k \ras{\sigma_k}_\FA s_{k+1}$ is not defined.
The {\em (accepted) language} corresponding to a DFA $\FA$ is
the set of accepted input words, which we denote by $\CA{L}(\FA)$.

\begin{definition}[Transition System, TS]
A transition system (TS) is a tuple $\TS = (X, x_0, \Delta, \AP, h)$, where:
\begin{itemize}
    \item $X$ is a finite set of states;
    \item $x_0 \in X$ is the initial state;
    \item $\Delta \subseteq X \times X$ is a set of transitions;
    \item $\AP$ is a set of properties (atomic propositions);
    \item $h : X \ra \spow{\Pi}$ is a labeling function.
\end{itemize}
\end{definition}

We also denote a transition $(x, x') \in \Delta$ by $x \ra_\TS x'$.
A \emph{trajectory} (or run) of the system is an infinite sequence of
states $\BF{x} = x_0 x_1 \ldots$ such that $x_k \ra_\TS x_{k+1}$ for all $k \geq 0$.
A state trajectory $\BF{x}$ generates an \emph{output trajectory} $\BF{o} = o_0 o_1 \ldots$,
where $o_k = h(x_k)$ for all $k \geq 0$.
The {\em (generated) language} corresponding to a TS $\TS$ is
the set of all generated output words, which we denote by $\CA{L}(\TS)$.

\section{Time Window Temporal Logic}
\label{sec:spec-twtl}
Time window temporal logic (TWTL) was first introduced in~\cite{vasile2014}
as a rich specification language for robotics applications.
TWTL formulae are able to capture temporal logic specifications
about the service time windows and their durations.
TWTL is a linear-time logic encoding sets of discrete-time sequences
with values in a finite alphabet.

%The semantics of TWTL is rich enough for a large class of robotic missions.
%The properties of TWTL language will also be discussed in Sec.~\ref{sec:props}.

A TWTL formula is defined over a set of atomic propositions $\AP$ and has the following syntax:
\begin{equation*}
\label{eq:logic-def}
\phi :: = H^d s \, | \, H^d \notltl s \, |
\, \phi_1 \andltl \phi_2 \, | \, \phi_1 \orltl \phi_2 \, | \, \notltl \phi_1
\, | \, \phi_1 \cdot \phi_2 \, | \, [\phi_1]^{[a, b]}
\end{equation*}
where $s$ is either the ``true'' constant $\True$
or an atomic proposition in $\AP$;
$\andltl$, $\orltl$, and $\notltl$ are the conjunction, disjunction, and negation
Boolean operators, respectively; $\cdot$ is the concatenation operator;
$H^d$ with $d \in \BB{Z}_{\geq 0}$ is the {\em hold} operator; 
and $[\ ]^{[a,b]}$ with $0 \leq a \leq b$ is the {\em within} operator.

The semantics of the operators is defined with respect to the
finite subsequences of a (possibly infinite) word $\BF{o}$ over $\spow{\AP}$.
Let $\BF{o}_{t_1, t_2}$ be the subsequence of $\BF{o}$,
which starts at time $t_1 \geq 0$ and ends at time $t_2 \geq t_1$.
The {\em hold} operator $H^d s$ specifies that $s \in \AP$
should be repeated for $d$ time units.
The semantics of $H^d \notltl s$ is defined similarly,
but for $d$ time units only symbols from $\AP\setminus \{s\}$ should appear.
For convenience, if $d=0$ we simply write $s$ and $\notltl s$
instead of $H^0 s$ and $H^0 \notltl s$, respectively.
The word $\BF{o}_{t_1, t_2}$ satisfies $\phi_1 \andltl \phi_2$,
$\phi_1 \orltl \phi_2$, or $\notltl \phi$
if $\BF{o}_{t_1, t_2}$ satisfies both formulae, at least one formula,
or does not satisfy the formula, respectively.
The {\em within} operator $[\phi]^{[a, b]}$ bounds
the satisfaction of $\phi$ to the time window $[a, b]$.
The concatenation operator $\phi_1 \cdot \phi_2$ specifies that
first $\phi_1$ must be satisfied, and then immediately
$\phi_2$ must be satisfied.

Formally, the semantics of TWTL formulae is defined recursively as follows:
{\small
\begin{alignat*}{2}
&\BF{o}_{t_1, t_2} \models H^d s  & \text{ iff } &
    s \in o_t,  \forall t \in \{t_1, \ldots, t_1 + d\} \wedge (t_2-t_1 \geq d)\\
&\BF{o}_{t_1, t_2} \models H^d \notltl s & \text{ iff } &
    s \notin o_t,  \forall t \in \{t_1, \ldots, t_1 + d\} \wedge (t_2-t_1 \geq d)\\
&\BF{o}_{t_1, t_2} \models \phi_1 \andltl \phi_2 & \text{ iff } &
    (\BF{o}_{t_1, t_2} \models \phi_1) \andltl (\BF{o}_{t_1, t_2} \models \phi_2)\\
&\BF{o}_{t_1, t_2} \models \phi_1 \orltl \phi_2 & \text{ iff } &
    (\BF{o}_{t_1, t_2} \models \phi_1) \orltl (\BF{o}_{t_1, t_2} \models \phi_2)\\
&\BF{o}_{t_1, t_2} \models \notltl \phi & \text{ iff } &
    \notltl (\BF{o}_{t_1, t_2} \models \phi)\\
&\BF{o}_{t_1, t_2} \models \phi_1 \cdot \phi_2 & \text{ iff } &
(\exists t = {\arg\min}_{t_1\leq t < t_2} \{\BF{o}_{t_1, t} \models \phi_1 \}) \andltl\\
& & & \left( \BF{o}_{t+1, t_2} \models \phi_2 \right)\\
&\BF{o}_{t_1, t_2} \models [\phi]^{[a, b]} & \text{ iff } &
\exists t \geq t_1 + a \text{ s.t. } \BF{o}_{t, t_1+b} \models \phi  \wedge (t_2-t_1 \geq b)
\end{alignat*}
}
%
%\noindent where we have the {\em length constraints}
%$d \leq t_2-t_1$ and $b \leq t_2 - t_1$.
A word $\BF{o}$ is said to satisfy a formula $\phi$
if and only if there exists $T \in \{0,\ldots, \card{\BF{o}}\}$ such that
$\BF{o}_{0,T} \models \phi$.
%Note that if any of the length constraints are violated,
%then the satisfaction of the formula can not be decided.
%Therefore, the length of the word $\BF{o}$ must be large enough.

%Checking a word against a TWTL formula can be performed
%in bounded time. 
%The semantics of TWTL imply that the satisfaction
%of a TWTL formula can be decided within bounded time.
A TWTL formula $\phi$ can be verified with respect to a bounded word. Accordingly, we define the {\em time bound} of $\phi$, i.e., $\norm{\phi}$, as the maximum time needed to satisfy $\phi$, which can be recursively computed as follows:

{\small
\begin{equation}
\label{eq:bound-def}
\norm{\phi} = \begin{cases}
\max(\norm{\phi_1}, \norm{\phi_2}) & \mbox{if } \phi \in \{ \phi_1 \andltl \phi_2, \phi_1 \orltl \phi_2 \} \\
\norm{\phi_1} & \mbox{if } \phi = \notltl \phi_1 \\
\norm{\phi_1} + \norm{\phi_2} + 1 & \mbox{if } \phi = \phi_1 \cdot \phi_2\\
d & \mbox{if } \phi \in \{ H^d s, H^d \notltl s \} \\
b & \mbox{if } \phi = [\phi_1]^{[a, b]}
\end{cases}
\end{equation}
}

We denote the language of all words satisfying $\phi$ by $\CA{L}(\phi)$.
Note that TWTL formulae are used to specify
prefix languages of either $\Sigma^*$ or $\Sigma^\omega$,
where $\Sigma = \spow{\AP}$.
Moreover, the number of operators in a TWTL formula $\phi$
is denoted by $\card{\phi}$. 

Some examples of TWTL formulae
%\textcolor{blue}{
for a robot servicing at some regions
%}
can be as follows:

\noindent - {\it servicing within a deadline:} ``service $A$ for 2 time units before 10'', 
\begin{equation}
\phi_1 = [H^2 A]^{[0, 10]} \text{ and } \norm{\phi_1} = 10.
\label{eq:ex1}
\end{equation}

\noindent - {\it servicing within time windows:} ``service $A$ for 4 time units within [3, 8]
and $B$ for 2 time units within [4, 7]'',
\begin{equation}
\phi_2 = [H^4 A]^{[3, 8]} \andltl [H^2 B]^{[4, 7]} \text{ and } \norm{\phi_2} = 8.
\label{eq:ex2}
\end{equation}

\noindent - {\it servicing in sequence:} ``service $A$ for 3 time units within [0, 5] and after this service 
$B$ for 2 time units within [4, 9]'',
\begin{equation}
\phi_3 = [H^3 A]^{[0,5]} \cdot [H^2 B]^{[4,9]} \text{ and } \norm{\phi_3} = 15.
\label{eq:ex3}
\end{equation}

\noindent - {\it enabling conditions:} ``if $A$ is serviced for 2 time units within 9 time units, then $B$ should be serviced for 3 time units within the same time interval (i.e., within 9 time units)'',
\begin{equation}
\phi_4 = [H^2 A \Implies [H^3 B]^{[2,5]}]^{[0,9]} \text{ and } \norm{\phi_4} = 9,
\label{eq:ex4}
\end{equation}
where $\Implies$ denotes implication.

TWTL provides some benefits over other time-bounded temporal logics.
%\textcolor{blue}{
One of the main benefits of TWTL is the existence of an explicit concatenation operator, which results in compact representation of serial tasks. For instance, the specification in \eqref{eq:ex3} is expressed in TWTL, BLTL, and MTL in Table~\ref{table1}, where the MTL formula contains a set of recursively defined sub-formulae connected by disjunctions whereas the BLTL formula contains nested temporal operators with conjunction. In both cases, dealing with the disjunction of numerous sub-formulae or the nested temporal operators with conjunction significantly increases the complexity of constructing the automaton (i.e., in exponential or quadratic way, respectively \cite{Maia13}). On the other hand, stemming from the compact representation of TWTL, we provide a linear-time algorithm to handle the concatenations of tasks under some mild assumptions.
%}

\begin{table}[h]
\caption{The representation of \eqref{eq:ex3} in TWTL, BLTL, and MTL.}
\label{table1}
\center
\begin{tabular}{| l | l |}
 \hline
 TWTL & $ [H^3 A]^{[0,5]} \cdot [H^2 B]^{[4,9]}$ \\
 \hline
 BLTL & $\Event^{\leq 5-3} ( \Always^{\leq 3} A \andltl \Event^{\leq 9-2+3} \Always^{\leq 2} B)$ \\
 \hline
 MTL & $\bigvee_{i=0}^{5-3}(\Always_{[i,i+3]} A \andltl \bigvee_{j=i+3+4}^{i+3+9-2} \Always_{[j,j+2]} B)$ \\
 \hline
\end{tabular}
\end{table}

%\textcolor{blue}{
In addition to the concatenation operator, the existence of within and hold operators also leads to compact (shorter length) representation of specifications, which greatly improves the readability of the formula. For example, the specification in \eqref{eq:ex2} is expressed in various temporal logics in Table~\ref{table2} where the BLTL formula contains nested temporal operators with shifted time windows whereas the MTL formula consists of the disjunction of many sub-formulae.
%}

%in TWTL indicates the exact time bounds of a desired task whereas the temporal operators in other temporal logics (e.g., MTL or STL) has a shifting time notion. The exact representation of temporal properties particularly improves the readability of the specifications written in TWTL. Second, TWTL can express specifications in a compact fashion, which mainly stems from the existence of concatenation, hold, and within operators. For instance, Table~\ref{tab:TLcomparison} illustrates the representations of the formulae in \eqref{eq:ex2} and \eqref{eq:ex3} in terms of TWTL, BLTL, and MTL (with a fragment from~\cite{Karaman.Frazzoli:CDC08}). 
\begin{table}[h]
\caption{The representation of \eqref{eq:ex2} in TWTL, BLTL, and MTL.}
\label{table2}
\center
\begin{tabular}{| l | l |}
\hline
 TWTL & $[H^4 A]^{[3, 8]} \andltl [H^2 B]^{[4, 7]}$ \\
 \hline
 BLTL & $\Event^{\leq 8-4} \Always^{\leq 4} A \andltl \Event^{\leq 7-2} \Always^{\leq 2} B$\\
 \hline
 MTL & $\bigvee_{i=3}^{8-4} \Always_{[i, i+4]} A \andltl \bigvee_{i=4}^{7-2} \Always_{[i,i+2]} B$ \\
 \hline 
\end{tabular}
\end{table}

%\textcolor{blue}{
For automata-based model-checking, a BLTL formula is translated into another off-the-shelf temporal logic (e.g., syntactically co-safe linear temporal logic (scLTL) \cite{kupferman2001}), for which an existing tool (e.g., {\em scheck}~\cite{Latvala03}) for the automaton construction can be used \cite{Tkachev13}.
%}
On the other hand, MTL and STL are very expressive temporal logics that are particularly used for real-time systems. While there is no finite representation for the satisfying language of STL, timed-automata \cite{alur1994} are used to represent the satisfying language of MTL. Compared to the other temporal logics, TWTL has a significantly lower computational complexity since an automaton for the satisfying language of a TWTL formula can be constructed directly
%\textcolor{blue}{
(see Sec.~\ref{sec:construction})
%}
and does not require any clocks to deal with the time constraints (as in timed automata).
%\textcolor{blue}{
Finally, for a given TWTL formula $\phi$, we show that all possible temporally relaxed $\phi$ can be encoded to a very compact representation, which is enabled from the definition of temporal relaxation introduced in the next section.

\section{Temporal Relaxation}
\label{sec:tr}

In this section, we introduce a {\em temporal relaxation}
of a TWTL formula.
This notion is used in Sec.~\ref{sec:pb-formulation} to formulate
an optimization problem over temporal relaxations.

To illustrate the concept of temporal relaxation,
consider the following TWTL formula:
\begin{equation}
\label{eq:tr-example-twtl}
\phi_1 = [H^{1} A]^{[0:2]} \cdot \big[H^{3} B \andltl [H^{2} C]^{[0:4]}\big]^{[1:8]}.
\end{equation}

In cases where $\phi_1$ cannot be satisfied, one question is:
what is the ``closest" achievable formula that can be performed?
Hence, we investigate relaxed versions of $\phi_1$.
One way to do this is to relax the deadlines for the time windows,
which are captured by the {\em within} operator.
Accordingly, a relaxed version of $\phi_1$ can be written as
{\small
\begin{equation}
\label{eq:tr-example-twtl-relaxed}
\phi_1(\BS{\tau}) = [H^{1} A]^{[0:(2+\tau_1)]} \cdot [H^{3} B \andltl [H^{2} C]^{[0:(4+\tau_2)]}]^{[1:(8+\tau_3)]},
\end{equation}}%
where $\BS{\tau}=(\tau_1, \tau_2, \tau_3) \in \BB{Z}^3$.
Note that a critical aspect while relaxing the time windows is
to preserve the feasibility of the formula.
This means that all sub-formulae of $\phi$ enclosed by the
{\em within} operators must take less time to satisfy
than their corresponding time window durations.

%This means that any sub-task expressed as a sub-formula in $\phi$
%should be achievable in the given time window.

\begin{definition}[Feasible TWTL formula]
A TWTL formula $\phi$ is called feasible,
if the time window corresponding to each {\em within} operator
is greater than the duration of the corresponding enclosed task
(expressed via the {\em hold} operators).
\end{definition}

\begin{remark}
Consider the formula in~Eq.\eqref{eq:tr-example-twtl-relaxed}.
For $\phi_1(\BS{\tau})$ to be a feasible TWTL formula,
the following constraint must hold:
(i) $2+\tau_1 \geq 1$; (ii) $4+\tau_2 \geq 2$ and
(iii) $7+\tau_3 \geq \max\{3, 4+\tau_2\}$.
Note that $\BS{\tau}$ may be non-positive.
In such cases, $\phi_1(\BS{\tau})$ becomes a stronger specification than $\phi_1$,
which implies that the sub-tasks are performed ahead of their actual deadlines. 
\end{remark}

Let $\phi$ be a TWTL formula. Then, a $\tau-$relaxation of $\phi$ is defined as follows:
\begin{definition}[$\tau-$Relaxation of $\phi$]
\label{def:eps-relax}
Let $\mbox{\boldmath$\tau$}\in \BB{Z}^{m}$, where $m$ is the number of {\em within} operators contained in $\phi$. The $\tau$-relaxation of $\phi$ is a {\em feasible} TWTL formula $\phi(\mbox{\boldmath$\tau$})$, where each subformula of the form $[\phi_i]^{[a_i, b_i]}$ is replaced by $[\phi_i]^{[a_i, b_i + \tau_i]}$.% and either $\tau\geq 0$ or $\tau \leq 0$.
\end{definition}
\begin{remark}
For any $\phi$, $\phi(\bf{0}) = \phi$. 
\end{remark}

\begin{definition}[Temporal Relaxation]
\label{def:temporal-relaxation}
Given $\phi$, let $\phi(\BS{\tau})$ be a feasible relaxed formula.
The temporal relaxation of $\phi(\BS{\tau})$
is defined as $\trel{\BS{\tau}} =\max_j (\tau_j)$.
\end{definition}

\begin{remark}
\label{rem:optimal2satisfaction}
If a word $o \models \phi(\BS{\tau})$ with $\trel{\BS{\tau}}  \leq 0$,
then $o \models \phi$.
\end{remark}

\section{Problem formulation}
\label{sec:pb-formulation}
In this section, first, we propose a generic optimization problem over
temporal relaxations of a TWTL formula. Then, we show how this setup can
be used to formulate verification, synthesis, and learning problems.

The objective of the following optimization problem is to find a
feasible relaxed version of a TWTL formula that optimizes a
cost function penalizing the sets of satisfying and unsatisfying
words, and the vector of relaxations.

\begin{problem}
\label{pb:general}
Let $\phi$ be a TWTL formula over the set of atomic propositions $\AP$, and let
$\CA{L}_1$ and $\CA{L}_2$ be any two languages over the alphabet $\Sigma=\spow{\AP}$.
Consider a cost function
$F : \BB{Z}_{\geq 0} \times \BB{Z}_{\geq 0} \times \BB{Z}^m \ra \overline{\BB{R}}$,
where $m$ is the number of {\em within} operators contained in $\phi$.
Find $\BS{\tau}$ such that
$F(\card{\CA{L}(\phi(\BS{\tau})) \cap \CA{L}_1},
\card{\CA{L}(\notltl\phi(\BS{\tau})) \cap \CA{L}_2}, \BS{\tau})$
is minimized.
\end{problem}

\subsection{Verification, synthesis, and learning}
\label{sec:problems}

%The general problem formulation Pb.~\ref{pb:general} can be
%specialized to formulate verification, synthesis, and learning problems.
%In the following we present three possible formulations.

In the following, we use Problem~\ref{pb:general} to formulate
three problems for verification, synthesis, and learning.

\subsubsection{Verification}\footnote{
This problem is not a verification problem in the usual sense,
but rather finding a formula that is satisfied by all runs of a system.}
Given a transition system $\TS$
and a TWTL formula $\phi$, we want to check if there exists
a relaxed formula $\phi(\BS{\tau})$ such that all output words
generated by $\TS$ satisfy $\phi(\BS{\tau})$.

In Problem~\ref{pb:general}, we can set $\CA{L}_1 = \emptyset$
and $\CA{L}_2 = \CA{L}(\TS)$, and we choose the following cost function:
\begin{equation}
\label{eq:pb-formulation-verification}
F(x, y, \BS{\tau}) = 1 - \delta(y),
\end{equation}
where $x, y \in \BB{Z}_{\geq 0}$ and $\delta(x) = \begin{cases}1 & x=0\\ 0 & x\neq 0\end{cases}$. % is the Dirac function.
The cost function in Eq.~\eqref{eq:pb-formulation-verification}
has a single global minimum value at 0 which corresponds
to the case $\CA{L}(\TS) \cap \CA{L}(\notltl \phi(\BS{\tau})) = \emptyset$.

\subsubsection{Synthesis}
Given a transition system $\TS$
and a TWTL formula $\phi$, we want to find
a policy (a trajectory of $\TS$) that produces
an output word satisfying a relaxed version $\phi(\BS{\tau})$
of the specification with minimal temporal relaxation $\trel{\BS{\tau}}$.
%
%the relaxed formula $\phi(\BS{\tau})$ with minimum
%relaxation $\trel{\BS{\tau}}$ such that there
%exists a trajectory of $\TS$ whose corresponding
%output words satisfies $\phi(\BS{\tau})$.

In Problem~\ref{pb:general}, we can set $\CA{L}_1 = \CA{L}(\TS)$
and $\CA{L}_2 = \emptyset$, and we choose the following cost function:
\begin{equation}
\label{eq:pb-formulation-synthesis}
F(x, y, \BS{\tau}) = \begin{cases}
\trel{\BS{\tau}} & x > 0\\
\infty & \text{otherwise}
\end{cases},
\end{equation}
where $x, y \in \BB{Z}_{\geq 0}$.
The cost function in Eq.~\eqref{eq:pb-formulation-synthesis}
is minimized by an output word of $\TS$, which satisfies the
relaxed version of $\phi$ with minimum temporal relaxation,
see Def.~\ref{def:temporal-relaxation}.

\subsubsection{Learning}
Let $\phi$ be a TWTL formula and $\CA{L}_p$ and $\CA{L}_n$ be
two finite sets of words labeled as positive and negative examples,
respectively.
We want to find a relaxed formula $\phi(\BS{\tau})$
such that the misclassification rate, i.e.,
$\card{\{w \in \CA{L}_p\ |\ w\not\models \phi(\BS{\tau}) \} } + \card{\{w\in \CA{L}_n \ |\ w \models \phi(\BS{\tau}) \}}$,
is minimized.

This case can be mapped to the generic formulation
by setting $\CA{L}_1 = \CA{L}_n$, $\CA{L}_2 = \CA{L}_p$
and choosing the cost function
\begin{equation}
\label{eq:pb-formulation-learning}
F(x, y, \BS{\tau}) = x+y,
\end{equation}
which captures the misclassification rate,
where $x, y \in \BB{Z}_{\geq 0}$.

\subsection{Overview of the solution}

We propose an automata-based approach to solve the verification,
synthesis, and learning problems defined above.
Specifically, the proposed algorithm constructs an annotated DFA
$\FA_\infty$, which captures all temporal relaxations of the given
formula $\phi$, i.e., $\CA{L}(\FA_\infty) = \CA{L}(\phi(\infty))$
%The solution is based on a recursive algorithm to construct
%a DFA from a TWTL formula $\phi$. The algorithm may be used
%to construct an automaton $\FA$ corresponding to $\CA{\phi}$,
%which captures the language of satisfying words,
%i.e. $\CA{L}(\FA) = \CA{L}(\phi)$.
%However, one of the main contributions of the paper is that
%we provide an algorithm to construct an annotated DFA $\FA_\infty$,
%which captures all temporal relaxations of the given formula $\phi$,
%i.e. $\CA{L}(\FA_\infty) = \CA{L}(\phi(\infty))$
(see Def.~\ref{def:phi-infty} for the definition of $\phi(\infty)$).
Note that the algorithm can also be used to construct a (normal) DFA
$\FA$ which accepts the satisfying language of $\phi$, i.e.,
$\CA{L}(\FA) = \CA{L}(\phi)$.
Using the resulting DFA $\FA_\infty$, we proceed in Sec.~\ref{sec:algorithms}
to solve the synthesis and verification problems using a product automaton
approach. For the synthesis problem, we propose a recursive
algorithm that computes a satisfying path with minimum temporal relaxation.
%{\color{orange}
The learning problem is solved by inferring the minimum relaxation
for each trajectory and then combining these relaxations to ensure minimum
misclassification rate.
%}

\section{Properties of TWTL}
\label{sec:props}
In this section, we present properties of TWTL formulae, their temporal relaxations,
and their accepted languages.

%\subsection{TWTL}
%\label{sec:twtl-props}

%In this paper, languages are interpreted with respect to
%three representations: TWTL formulae, automata, and sets.
In this paper, languages are represented in three ways: as
TWTL formulae, as automata, and as sets.
As one might expect, there is a duality between some
operators of TWTL and set operations, i.e., conjunction,
disjunction, and concatenation correspond to intersection,
union, and concatenation languages, respectively.
Negation may be mapped to complementation with respect to
the language of all bounded words, where the bound
is given by the time bound of the negated formula.

\begin{proposition}
\label{th:twtl-props}
The following properties hold
\begin{align}
(\phi_1 \cdot \phi_2) \cdot \phi_3 &= \phi_1 \cdot ( \phi_2 \cdot \phi_3) \label{eq:twtl-concat-assoc}\\
\phi_1 \cdot ( \phi_2 \orltl \phi_3) &= (\phi_1 \cdot \phi_2) \orltl (\phi_1 \cdot \phi_3) \label{eq:twtl-concat-distrib-disj}\\
[\phi_1 \orltl \phi_2]^{[a, b]} &= [\phi_1]^{[a, b]} \orltl [\phi_2]^{[a, b]} \label{eq:twtl-within-distrib-disj}\\
\notltl (H^d p) &= [\notltl p]^{[0, d]}\\
[\phi_1]^{[a_1, b_1]} &= (H^{a_1-1} \True) \cdot [\phi_1]^{[0, b_1-a_1]}\\
(H^{d_1} p) \cdot (H^{d_2} p) &= H^{d_1+d_2+1} p\\
[\phi_1]^{[a, b]} &\Implies [\phi_1]^{[a, b + \tau]} \label{eq:twtl-prop-tw-expand}\\
(\phi_1 \Implies \phi_2) &\Implies ([\phi_1]^{[a, b]} \Implies [\phi_2]^{[a, b]}) \label{eq:twtl-prop-within-implication}
\end{align}
where $\phi_1$, $\phi_2$, and $\phi_3$ are TWTL formulae,
$p \in \{s, \notltl s\}$, $s \in \AP \cup \{\True\}$,
and $a, b, a_1, b_1, d, d_1, d_2, \tau \in \BB{Z}_{\geq 0}$
such that $a \leq b$ and $1 \leq a_1\leq b_1$.
\end{proposition}
\begin{proof}
These follow directly from the semantics of TWTL formulae.
\end{proof}

\begin{definition}[Disjunction-Free Within form]
Let $\phi$ be a TWTL formula. We say that $\phi$ is in
Disjunction-Free Within (DFW) form if for all {\em within}
operators contained in the formula the associated
enclosed subformulae do not contain any disjunction
operators.
\end{definition}

An example of a TWTL formula in DFW form is
$\phi_1= [H^2 A]^{[0,9]} \orltl [H^5 B]^{[0,9]}$,
while a formula not in DFW form is
$\phi_2 = [H^2 A \orltl H^5 B]^{[0,9]}$.
However, $\phi_1$ and $\phi_2$ are equivalent
by Eq.~\eqref{eq:twtl-within-distrib-disj} of Prop.~\ref{th:twtl-props}.
The next proposition formalizes this property.

\begin{proposition}
\label{th:wdf}
%Any TWTL formula $\phi$ with negation operators
%only in front of atomic propositions can be put in DFW form.
For any TWTL formula $\phi$, if the negation operators are
only in front of the atomic propositions, then $\phi$ can be
written in the DFW form. 
\end{proposition}
\begin{proof}
The result follows from the properties of distributivity of
Boolean operators and Prop.~\ref{th:twtl-props}, which
can be applied iteratively to move all disjunction operators
outside the {\em within} operators.
%If negation is only in front of atomic propositions,
%then it is ensured that $\phi$ does not
%contain hidden disjunctions in the formula, which might
%results by applying DeMorgan's laws.
\end{proof}

In the following, we define the notion of unambiguous
concatenation, which enables tracking of
progress for sequential specifications.
Specifically, if the property holds, then
an algorithm is able
to decide at each moment if the first
specification has finished while monitoring
the satisfaction of two sequential specifications. 

\begin{definition}
\label{def:unambiguous-concat}
Let $\CA{L}_1$ and $\CA{L}_2$ be two languages.
We say that the language $\CA{L}_1 \cdot \CA{L}_2$
is an {\em unambiguous concatenation} if each word in the resulting language
can be split unambiguously,
i.e., $\big( L_1, \CA{L}_1, \CA{L}_1 \cdot ( \prefix{\CA{L}_2} \setminus \{\epsilon\}) \big)$
is a partition of $\prefix{\CA{L}_1 \cdot \CA{L}_2}$,
where $L_1 = \{ w_{0, i} \ |\ w\in \CA{L}_1, i\in \{0, \ldots, \card{w}-2 \} \}$
and $P(L)$ denotes the maximal prefix language of $L$.
\end{definition}

The three sets of the partition from
Def.~\ref{def:unambiguous-concat}
may be thought as indicating whether
the first specification is in progress,
the first specification has finished, and
the second specification is in progress,
respectively.

\begin{proposition}
\label{th:unambiguous-concat}
Consider two languages $\CA{L}_1$ and $\CA{L}_2$.
The language $\CA{L}_1 \cdot \CA{L}_2$ 
is an unambiguous concatenation
if and only if $\CA{L}_1$ is an unambiguous language.
\end{proposition}
\begin{proof}
See App.~\ref{app:unambiguous-concat}
%Let $\big( L_1, \CA{L}_1, \CA{L}_1 \cdot ( \prefix{\CA{L}_2} \setminus \{\epsilon\}) \big)$
%be a partition of $\prefix{\CA{L}_1 \cdot \CA{L}_2}$ and
%$L$ be a proper subset of $\CA{L}_1$.
%Assume that there exists $w \in L$ and
%$w' \in \CA{L}_1 \setminus L$ such that
%$w=w'_{0, i}$, for some $i \in \{0, \ldots, \card{w'}-1\}$.
%It follows that $w \in L_1$, because $w\neq w'$.
%However, this contradicts the fact that $L_1$ and
%$\CA{L}_1$ are disjoint.
%
%Conversely, let $\CA{L}_1$ be unambiguous and
%consider a word $w \in \prefix{\CA{L}_1 \cdot \CA{L}_2}$.
%Assume that $w \in L_1 \cap \CA{L}_1$.
%It follows that $\{w\}$ is a prefix language
%for $\CA{L}_1\setminus \{w\}$, which contradicts with the
%hypothesis that $\CA{L}_1$ is unambiguous.
%Similarly, if we assume that there exists
%$w \in \prefix{\CA{L}_1} \cap \big(\CA{L}_1\cdot (\prefix{\CA{L}_2} \setminus \{\epsilon\}) \big)$,
%then there exists $w', w'' \in \CA{L}_1$ such that
%$w'$ is a prefix of $w$, $w$ is a prefix of $w''$,
%and $\card{w'} < \card{w} \leq \card{w''}$.
%Thus, we arrive again at a contradiction with
%the unambiguity of $\CA{L}_1$. Thus, the three
%sets form a partition of $\prefix{\CA{L}_1 \cdot \CA{L}_2}$.
\end{proof}

%\subsection{Temporal Relaxation}
%\label{sec:tr-props}

%In this section, we present some important
%properties of temporal relaxations of TWTL
%formulae.
In the following results, we frequently use
the notion of abstract syntax tree of a TWTL formula.

\begin{definition}
\label{def:ast}
An {\em Abstract Syntax Tree (AST)} of $\phi$ is denoted by $AST(\phi)$,
where each leaf corresponds to a {\em hold} operator
and each intermediate node corresponds to a Boolean,
concatenation, or {\em within} operator.
\end{definition}

Given a TWTL formula $\phi$, there might exist
multiple AST trees that represent $\phi$.
In this paper, $AST(\phi)$ is
assumed to be computed by an LL(*) parser~\cite{Parr07}.
The reader is referred to \cite{Hopcroft2006} for more details on AST and parsers.
An example of an AST tree of Eq.~\eqref{eq:tr-example-twtl}
is illustrated in Fig.~\ref{fig:ast-example}.

%\begin{figure}[!htb]
%  \centering
%  \includegraphics[trim =0mm 0mm 0mm 0mm, clip,width=0.8\linewidth]{ast.pdf}
%  \caption{The AST corresponding to the TWTL in Eq.~\eqref{eq:tr-example-twtl}.}
%  \label{fig:ast-example}
%\end{figure}

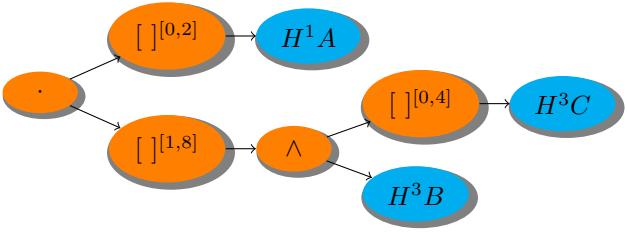
\begin{figure}[!htb]
\centering
\begin{tikzpicture}[
    grow = right,
    state/.style={ellipse, draw=none, fill=orange, circular drop shadow,
        text centered, anchor=west, text=black, minimum width=1cm},
    leaf/.style={ellipse, draw=none, fill=cyan, circular drop shadow,
        text centered, anchor=west, text=black},
    level distance=0.4cm, growth parent anchor=east
]

\node [state] {$\cdot$} [->, sibling distance=1.5cm]
     child {
        node [state] {$[\ ]^{[1,8]}$}
        child { [->, sibling distance=1.2cm]
            node [state] {$\andltl$}
            child {
                node [leaf] {$H^3 B$}
            }
            child {
                node [state] {$[\ ]^{[0, 4]}$}
                child {
                    node [leaf] {$H^3 C$}
                }
            }
        }
     }
     child {
        node [state] {$[\ ]^{[0,2]}$}
        child{
            node [leaf] {$H^1 A$}
        }
     }
;
\end{tikzpicture}
\caption{An AST corresponding to the TWTL in Eq.~\eqref{eq:tr-example-twtl}.
The intermediate orange nodes correspond to the Boolean,
concatenation, and {\em within} operators,
while the cyan leaf nodes represent the {\em hold} operators.
%Leaf nodes are shown in cyan and are associated with {\em hold} operators.
}
\label{fig:ast-example}
\end{figure}

\begin{proposition}
\label{th:tr-partial-order}
Let $\BS{\tau}', \BS{\tau}'' \in \BB{Z}^m$ such that
$\phi(\BS{\tau}')$ and $\phi(\BS{\tau}'')$ are two feasible relaxed formulae,
where $m$ is the number of {\em within} operators in $\phi$.
If $\BS{\tau}' \leq \BS{\tau}''$, then $\phi(\BS{\tau}') \Implies \phi(\BS{\tau}'')$.
\end{proposition}
\begin{proof}
See App.~\ref{app:tr-partial-order}
%The proof follows by structural induction over $AST(\phi)$.
%The base case is trivial, since the leafs correspond to the {\em hold} operators.
%For the induction step, the result follows trivially if the
%intermediate node is associated with a Boolean or concatenation
%operator. The case of a {\em within} operator follows from
%Eq.~\eqref{eq:twtl-prop-tw-expand} and~\eqref{eq:twtl-prop-within-implication}
%in Prop.~\ref{th:twtl-props}, i.e.
%$[\phi(\BS{\tau})]^{[a, b + \tau_1]} \Implies [\phi(\BS{\tau}')]^{[a, b + \tau_1]} \Implies [\phi(\BS{\tau}')]^{[a, b + \tau'_1]}$,
%where $a < b \in \BB{Z}_{\geq 0}$ and $\BS{\tau} \leq \BS{\tau}' \in \BB{Z}^m$.
%We assumed without loss of generality that the first component of
%the temporal relaxation vectors is assigned to the root node.
\end{proof}

\begin{definition}
\label{def:phi-infty}
Given an output word $\BF{o}$, we say that $\BF{o}$
satisfies $\phi(\infty)$, i.e., $\BF{o} \models \phi(\infty)$,
if and only if
$\exists \BS{\tau}' < \infty \; s.t. \; \BF{o} \models \phi(\BS{\tau}')$.
%\begin{equation}
% \exists \mbox{\boldmath$\tau'$} < \infty \; s.t. \; \BF{o} \models \phi(\mbox{\boldmath$\tau'$}).
% \end{equation}
\end{definition}

The next corollary follows directly from Prop.~\ref{th:tr-partial-order}.
\begin{corollary}
Let $\BS{\tau} < \infty$, then $\phi(\BS{\tau}) \Implies \phi(\infty)$, $\forall \BS{\tau}$.
\end{corollary}
 
\begin{proposition}
Let $\phi(\BS{\tau}')$ and $\phi(\BS{\tau}'')$ be two feasible relaxed formulae.
If $\BS{\tau}' \leq \BS{\tau}''$, then $\norm{\phi(\BS{\tau}')} \leq \norm{\phi(\BS{\tau}'')}$.
\end{proposition}
\begin{proof}
The result follows by structural induction from Eq.~\eqref{eq:bound-def}
using a similar argument as in the proof of Prop.~\ref{th:tr-partial-order},
see App.~\ref{app:tr-partial-order}
\end{proof}

%\subsection{Finite languages}

An important observation about TWTL
is that the accepted languages corresponding
to formulae are finite languages.
In the following, we characterize such languages
in terms of the associated automata.

\begin{definition}
A DFA is called {\em strict} if and only if
(i) the DFA is blocking,
(ii) all states reach a final state, and
(ii) all states are reachable from the initial state.
\end{definition}

\begin{proposition}
\label{th:convert2strict}
Any DFA $\FA$ may be converted to a strict DFA in $O(\card{S_\FA})$ time.
\end{proposition}
\begin{proof}
States unreachable from the initial state can be identified
by traversing the automaton graph from the initial state using
either breath- or depth-first search. Similarly, the states
%which do not reach a final state
not reaching a final state
can be removed by traversing the automaton graph
using the reverse direction of the transitions.
Both operations take at most $O(\card{\delta_\FA}) = O(\card{S_\FA})$,
since there are at most $\card{\Sigma}$ transitions outgoing from
each state, where $\Sigma$ is the alphabet of $\FA$.
\end{proof}

Note that a strict DFA is not necessarily minimal with respect to
the number of states.

\begin{proposition}
\label{th:finite-dag}
If $\CA{L}$ is a finite language over an alphabet $\Sigma$,
then the corresponding strict DFA is a directed acyclic graph (DAG).
Moreover, given a (general) DFA $\FA$, checking if its associated language
$\CA{L}(\FA)$ is finite takes $O(\card{S_\FA})$ time.
\end{proposition}
\begin{proof}
For the first part, assume for the sake of contradiction
that $\FA$ has a cycle. Then,
%We prove the first part of the proposition by contradiction.
%Assume that $\FA$ has a cycle, then
we can form words
in the accepted language by traversing the cycle $n \in \BB{Z}_{\geq 0}$
times before going to a final state. Note that the states in
the cycle are reachable from the initial state and also reach
a final state, because $\FA$ is a strict DFA. It follows that
$\CA{L}$ is infinite, which contradicts the hypothesis.
Checking if a DFA $\FA$ is DAG takes $O(\card{S_\FA})$
by using a topological sorting algorithm, because of the
same argument as in Prop.~\ref{th:convert2strict}.
\end{proof}

\begin{corollary}
\label{th:unambiguous-dfa-final}
Let $\CA{L}$ be a finite unambiguous language
over the alphabet $\Sigma$ and $\FA$
be its corresponding strict DFA.
The following two statements hold:
\begin{enumerate}
  \item if $s \in F_\FA$, then the set of outgoing transitions of $s$ is empty.
  \item $\FA$ may be converted to a DFA with only one final states.
\end{enumerate}
\end{corollary}
\begin{proof}
Consider a final state $s \in F_\FA$.
Assume that there exists $s' \in S_\FA$ such that
$s \ras{\sigma}_\FA s'$, where $\sigma \in \Sigma$.
Since $\FA$ is strict, it follows that there is another
final state $s'' \in F_\FA$ which can be reached from
$s'$. Next, we form the words $w$ and $w'$ leading to $s$
and $s''$ passing trough $s'$, respectively. Clearly, $w$ is
a prefix of $w'$, which implies that $\CA{L}$ is not an
unambiguous language.
The second statement follows from the first
by noting that in this case, merging all final states does not
change the accepted language of the DFA $\FA$.
\end{proof}

\section{Automata construction}
\label{sec:construction}

In this section, we present a recursive procedure
to construct DFAs for TWTL formulae and their
temporal relaxations. The resulting DFA are used
in Sec.~\ref{sec:algorithms} to solve the proposed
problems in Sec.~\ref{sec:problems}.

Throughout the paper, a TWTL formula is assumed to have the following properties: %In the following, we impose the following assumption.
\begin{assumption}
\label{assump:spec}
Let $\phi$ be a TWTL. Assume that
%(i) $\phi$ is given in DFW form,
(i) negation operators are only in front of atomic propositions,
and (ii) all sub-formulae of $\phi$ correspond to unambiguous languages.
\end{assumption}

%Note that (i) does not limit the expressiveness of the supported TWTL formulae,
%because by Prop.~\ref{th:wdf} any TWTL formula can be put in DFW form.
The second part (ii) of Assump.~\ref{assump:spec} is a desired property of specifications
in practice, because it is related to the tracking of progress towards the satisfaction
of the tasks.
More specifically, if (ii) holds, then the end of each sub-formula can be
determined unambiguously, i.e., without any look-ahead.

\subsection{Construction Algorithm}

In~\cite{vasile2014}, a TWTL formula $\phi$ is translated to
an equivalent scLTL formula, and then an off-the-shelf tool,
such as {\em scheck}~\cite{Latvala03} and {\em spot}~\cite{duret.13.atva},
is used to obtain the corresponding DFA.
In this paper, we propose an alternative construction,
shown in Alg.~\ref{alg:twtl2dfa}, with two main advantages:
($i$) the proposed algorithm is optimized for TWTL formulae
so it is significantly faster than the method used in~\cite{vasile2014},
and ($ii$) the same algorithm can be used to construct a special DFA,
which captures all $\tau$-relaxations of $\phi$,
i.e., the DFA $\FA_\infty$ corresponding to $\phi(\infty)$.

\begin{algorithm}
\caption{Translation algorithm -- $translate(\cdot)$}
\label{alg:twtl2dfa}
\DontPrintSemicolon
\KwIn{$\phi$ -- the specification as a TWTL formula in DFW form}
%\KwIn{{\it inf} -- a boolean specifying if the DFA for $\phi(\infty)$ should be computed}
\KwOut{$\FA$ -- translated DFA}
\BlankLine
\uIf{$\phi = \phi_1 \otimes \phi_2$, where $\otimes \in \{\andltl, \orltl, \cdot\}$}{
  $\FA_1 \asgn translate(\phi_1)$, $\FA_2 \asgn translate(\phi_2)$\;
  $\FA \asgn \varrho_{\otimes}(\FA_1, \FA_2)$\;
}
\uElseIf{$\phi = H^d p$, where $p \in \{s, \notltl s\}$ and $s \in \AP$}{
  $\FA \asgn \varrho_{H}(p, d, \AP)$
}
\ElseIf{$\phi = [\phi_1]^{[a, b]}$}{
  $\FA_1 \asgn translate(\phi_1)$\;
  \lIf{\it inf}{
    $\FA \asgn \varrho_{\infty}(\FA_1, a, b)$\;
  }\lElse{
    $\FA \asgn \varrho_{[\,]}(\FA_1, a, b)$\;
  }
}
\Return $\FA$
\end{algorithm}

Alg.~\ref{alg:twtl2dfa} constructs the DFA recursively
by traversing $AST(\phi)$ computed via an LL(*) parser~\cite{Hopcroft2006,Parr07}
from the leaves to the root.
If the parameter {\it inf} is true, then the returned DFA
is an annotated DFA $\FA_\infty$ corresponding to $\phi(\infty)$;
otherwise a normal DFA $\FA$ is returned.
Each operator has an associated algorithm $\varrho_\otimes$
with $\otimes \in \{\andltl, \orltl, \cdot, H, \infty, [\,]\}$,
which takes the DFAs corresponding to the operands
(subtrees of the operator node in the AST) as input.
Then, $\varrho_\otimes$ returns the DFA that accepts
the formula associated with the operator node.
In the following, we present elaborate on
%the operations corresponding to
all operators and related operations, such as
annotating a DFA, relabeling the states of a DFA,
or returning the truncated version of a DFA with
respect to some given bound.

\subsection{Annotation}
\label{sec:annotation}

The algorithms presented in this section use DFAs with some additional annotation.
In this subsection, we introduce a annotated DFA and two algorithms,
Alg.~\ref{alg:relabel} and Alg.~\ref{alg:relabel-tree}, that are used to (re)label
DFAs and the associated annotation data, respectively.

We assume the following conventions to simplify
the notation:
%(i) functions may be treated as relations, i.e., sets of tuples;
(i) there is a global boolean variable {\it inf} accessible by all algorithms, which
specifies whether the normal or the annotated DFAs are to be computed;
(ii) in all algorithms, we have $\Sigma=\spow{\AP}$;
(iii) an element of $\sigma \in \Sigma$ is called a {\em symbol}
and is also a set of atomic propositions, $\sigma \subseteq \AP$;
(iv) a symbol $\sigma$ is called {\em blocking} for a state $s$
if there is no outgoing transition from $s$ activated by $\sigma$.

\subsubsection{Annotation}

An {\em annotated DFA} is a tuple $\FA = (S_\FA, s_0, \Sigma, \delta, F_\FA, T_\FA)$,
where the first five components have the same meaning as in Def.~\ref{def:dfa} and
$T_\FA$ is a tree that corresponds to the AST of the formula associated with the DFA.
Each node $T$ of the tree contains the following information:
\begin{enumerate}
    \item $T.op$ is the operation corresponding to $T$;
    \item $T.I$ is the set of initial states of the automaton corresponding to $T$;
    \item $T.F$ is the set of final states of the automaton corresponding to $T$;
    \item $T.left$ and $T.right$ are the left and right child nodes of $T$, respectively.
\end{enumerate}
Additionally, if $T.op$ is $\orltl$ (disjunction), then $T$
has another attribute $T.choice$, which is explained in
Sec.~\ref{sec:op-conj-disj}.

Note that the associated trees are set to $\emptyset$
and are ignored, if the normal DFAs are computed, i.e.,
{\it inf} is false.

The labels of the states change during the construction of the automata. Alg.~\ref{alg:relabel-tree} is used
to update the labels stored in the data structures of
the tree. The algorithm takes the tree $T$ as input,
a mapping $m$ from the states to the new labels, and
a boolean value $e$ that specifies if the states are
mapped to multiple new states. The first step is to
convert the states' new labels to singleton sets if $e$ is false (line 1). Then, the algorithm proceeds
to process the tree recursively starting with $T$.
The mapping $m$ is then used to compute $t.I$
and $t.F$ by expanding each state to a set and then
computing the union of the corresponding sets (lines 5-6).
In the case of $op=\orltl$, the three sets $B$, $L$,
and $R$, which form the tuple $t.choices$ are also
processed. The elements of all three sets are pairs
of a state $s$ and a symbol $\sigma \in \Sigma$.
Alg.~\ref{alg:relabel-tree} converts the states of
all these pairs in the tree sets (lines 7-12).

\begin{algorithm}
\caption{$relabelTree(T, m, e)$}
\label{alg:relabel-tree}
\DontPrintSemicolon
\KwIn{$T$ -- a tree structure}
\KwIn{$m$ -- (complete) relabeling mapping}
\KwIn{$e$ -- boolean, true if $m$ maps states to sets of states}
\BlankLine

\lIf{$\neg e$}{
    $m(s) \asgn \{m(s)\}, \forall s$
}

$stack \asgn [T]$\;
\While{$stack\neq [\,]$}{
    $t \asgn stack.pop()$\;
    $t.I \asgn \bigcup_{s\in t.I}{m(s)}$\;
    $t.F \asgn \bigcup_{s\in t.F}{m(s)}$\;
    \If{$op=\orltl$}{
        $B, L, R \asgn t.choices$\;
        $B' \asgn \bigcup_{(s_B, \sigma) \in B} \{(s, \sigma) \ |\ s\in m(s_B)\}$\;
        $L' \asgn \bigcup_{(s_L, \sigma)\in L} \{(s, \sigma) \ |\ s\in m(s_L)\}$\;
        $R' \asgn \bigcup_{(s_R, \sigma) \in R} \{(s, \sigma) \ |\ s\in m(s_R)\}$\;
        $t.choices \asgn (B', L', R')$\;
    }
    \lIf{$t.left \neq \emptyset$}{
        $stack.push(t.left)$\;
    }
    \lIf{$t.right \neq \emptyset$}{
        $stack.push(t.right)$
    }
}
\end{algorithm}

\subsubsection{Relabeling a DFA}

The Alg.~\ref{alg:relabel} relabels the states of a DFA $\FA$
with labels given by the mapping $m$.
The map $m$ can be a partial function of the states.
The states not specified are labeled with integers
starting from $i_0$ in ascending order. If $m$ is empty, then all states
are relabeled with integers.
Lastly, if {\it inf} is true then the tree $T_\FA$ associated with the DFA
is also relabeled, otherwise it is set as empty.

\begin{algorithm}
\caption{$relabel(\FA, m, i_0)$}
\label{alg:relabel}
\DontPrintSemicolon
\KwIn{$\FA=(S_\FA, s_0, \Sigma, \delta, F_\FA)$ -- a DFA}
\KwIn{$m$ -- (partial) relabeling mapping}
\KwIn{$i_0$ -- start labeling index}
\KwOut{the relabeled DFA}
\BlankLine

\For{$s \in S_\FA$ s.t. $\nexists m(s)$}{
    $m(s) \asgn i_0$\;
    $i_0 \asgn i_0 + 1$\;
}
$S'_\FA \asgn \{m(s) \ |\ s \in S_\FA\}$\;
$\delta' \asgn \{m(s)\ras{\sigma}_\FA m(s') \ |\ s\ras{\sigma}_\FA s'\}$\;
$F'_\FA \asgn \{m(s) \ |\ s\in F_\FA\}$\;
\lIf{{\it inf}}{
    $T'_\FA \asgn relabelTree(T_\FA, m)$\;
}\lElse{
   $T'_\FA \asgn \emptyset$\;
}
\Return $(S'_\FA, m(s_0), \Sigma, \delta', F'_\FA, T'_\FA)$\;
\end{algorithm}

\subsection{Operators}

\subsubsection{Hold}
The DFA corresponding to a {\em hold} operator is constructed by Alg.~\ref{alg:hold}.
The algorithm takes as input an atomic proposition $s$ in positive or negative form,
a duration $d$, and the set of atomic propositions $\AP$.
The computed DFA has $d+2$ states (line 1) that are connected in series as
follows: (i) if $s$ is in positive form then the states are connected by all transitions
activated by symbols  which contain $s$ (lines 2-4);
and (ii) if $s$ is in negative form then the states are connected by all transitions
activated by symbols which do not contain $s$ (lines 5-7).
Lastly, if {\it inf} is true, a new leaf node is created (line 8).

\begin{algorithm}
\caption{$\varrho_H(p, d, \AP)$}
\label{alg:hold}
\DontPrintSemicolon
\KwIn{$p \in \{s, \notltl s\}$, $s \in \AP$}
\KwIn{$d$ -- hold duration}
\KwIn{$\AP$ -- set of atomic propositions}
\KwOut{DFA corresponding to $H^d p$}
\BlankLine
$S \asgn \{0, \ldots, d+1\}$\;
\eIf{$p=s$}{
   $\Sigma_ s \asgn \spow{AP} \setminus \spow{(\AP\setminus \{s\})}$\;
   $\delta \asgn \{ i \ras{\sigma}_\FA (i+1) \ |\ i\in \{0,\ldots, d\}, \sigma \in \Sigma_s \}$\;
}{
   $\Sigma_{\notltl s} \asgn \spow{(\AP\setminus \{s\})}$\;
   $\delta \asgn \{ i \ras{\sigma}_\FA (i+1) \ |\ i\in \{0,\ldots, d\}, \sigma \in \Sigma_{\notltl s} \}$\;
}
\lIf{\it inf}{
   $T \asgn tree(H^d, \emptyset, \emptyset, \{0\}, \{d+1\})$\;
}\lElse{
   $T \asgn \emptyset$\;
}
\Return $(S, 0, \spow{\AP}, \delta, \{d+1\}, T)$\;
\end{algorithm}

\subsubsection{Conjunction and disjunction}
\label{sec:op-conj-disj}

The construction for conjunction and disjunction operations is
based on the synchronous product construction and is similar
to the standard one~\cite{Hopcroft2006}.
However, $\varrho_\andltl$ and $\varrho_\orltl$ produce strict DFAs,
which only have one accepting state. Both algorithms recursively
construct the product automaton starting from the initial composite
state. In the following, we describe the details of the algorithms
separately.

{\it Conjunction:} The DFA corresponding to the conjunction operation
is constructed by Alg.~\ref{alg:conjunction}. The procedure is recursive
and the synchronization condition, i.e., the transition relation, is the
following: given two composite states $(s_1, s_2)$ and $(s'_1, s'_2)$,
there exists a transition from the first state to the second state if there
exists a symbol $\sigma$ such that:
(i) there exists pairwise transitions enabled by $\sigma$ in the two automata (lines 9-11),
i.e., $s_1 \ras{\sigma}_{\FA_1} s'_1$ and $s_2 \ras{\sigma}_{\FA_2} s'_2$;
(ii) one automaton reached a final state and the other has a
transition enabled by $\sigma$ (lines 5-8), i.e., either
(a) $s_1 = s'_1 = s_{f1}$ and $s_2 \ras{\sigma}_{\FA_2} s'_2$,
or (b) $s_1 \ras{\sigma}_{\FA_1} s'_1$ and $s_2 = s'_2 = s_{f2}$.
The first case covers the synchronous execution (simulation) of both $\FA_1$
and $\FA_2$ when a symbol is encountered.
The second case corresponds to the situation when the two automata
require words of different sizes to accept an input.
A simple example of this case is the DFA encoding $H^2 A \andltl H^3 B$
and the input word $\{A, B\}, \{A, B\}, \{A, B\}, \{B\}$, which clearly
satisfies the TWTL formula.

\begin{algorithm}
\caption{$\varrho_\andltl(\FA_1, \FA_2)$}
\label{alg:conjunction}
\DontPrintSemicolon
\KwIn{$\FA_1=(S_{\FA_1}, s_{01}, \Sigma, \delta_1, \{s_{f1}\}, T_{\FA_1})$ -- left DFA}
\KwIn{$\FA_2=(S_{\FA_2}, s_{02}, \Sigma, \delta_2, \{s_{f2}\}, T_{\FA_2})$ -- right DFA}
%\KwIn{{\it inf} -- boolean specifying DFAs are annotated}
\KwOut{DFA corresponding to $\CA{L}(\FA_1) \cap \CA{L}(\FA_2)$}
\BlankLine

$S \asgn \{(s_{01}, s_{02})\}$, $E \asgn \emptyset$\;
$stack \asgn [(s_{01}, s_{02})]$\;
\While{$stack \neq [\,]$}{
    $\BS{s}=(s_1, s_2) \asgn stack.pop()$\;
    \uIf{$s_1 = s_{f1}$}{
        $S_n \asgn \{((s_1, s'_2), \sigma) \ |\ s_2 \ras{\sigma}_{\FA_2} s'_2 \}$\;
    }
    \uElseIf{$s_2 = s_{f2}$}{
        $S_n \asgn \{((s'_1, s_2), \sigma) \ |\ s_1 \ras{\sigma}_{\FA_1} s'_1 \}$\;
    }
    \Else{
        $S_n \asgn \{((s'_1, s'_2), \sigma) \ |\ \exists\sigma\in\Sigma \text{ s.t. }$\;
        \qquad\qquad\qquad $(s_1 \ras{\sigma}_{\FA_1} s'_1) \wedge (s_2 \ras{\sigma}_{\FA_2} s'_2 )\}$\;
    }
    $E \asgn E \cup \{(\BS{s}, \sigma, \BS{s'}) \ |\ (\BS{s'}, \sigma) \in S_n \}$\;
    $S' \asgn \{\BS{s'} \ |\ \exists\sigma \in \Sigma \text{ s.t. } (\BS{s'}, \sigma) \in S_n \}$\;
    $stack.extends(S' \setminus S)$\;
    $S \asgn S \cup S'$\;
}
$m_L = \big\{(u, \{ (u, v) \in S_\FA \}) \ |\ u\in S_{\FA_1} \big\}$\;
$m_R = \big\{(v, \{ (u, v) \in S_\FA \}) \ |\ v\in S_{\FA_2} \big\}$\;
$T_\FA \asgn tree(\andltl, relabelTree(T_{\FA_1}, m_L, \top),$\;
    \quad $relabelTree(T_{\FA_2}, m_R, \top), \{(s_{01}, s_{02})\}, \{(s_{f1}, s_{f2})\})$\;
$\FA \asgn (S, (s_{01}, s_{02}), \Sigma, E, \{(s_{f1}, s_{f2})\}, T_\FA)$\;
\Return $relabel(\FA, \emptyset, 0)$
\end{algorithm}

Note that Alg.~\ref{alg:conjunction} generates only composite states which are
reachable from the initial composite state $(s_{01}, s_{02})$.
The resulting automaton has a single final state $(s_{f1}, s_{f2})$ which captures
the fact that both automata must accept the input word in order for the product
automaton to accept it.

After the automaton is constructed, the corresponding tree is created (lines 16-19).
The child subtrees are taken from $\FA_1$ and $\FA_2$, and relabeled.
The relabeling mapping expands each state $s$ to the set of all composite states,
which have $s$ as the first or second component corresponding to
whether $s$ is a state of the left or right automaton, respectively.

{\it Disjunction:} The disjunction operations is translated using
Alg.~\ref{alg:disjunction}. The first step of the algorithm is to
add a trap state in each of the two automata $\FA_1$ and $\FA_2$ (line 1).
All states of an automaton, except the final state, are connected via
blocking symbols to the trap state $\trap$ (lines 3-4). The trap state has
self-transitions for all symbols.
Afterwards, the algorithm creates the synchronous product automaton in
the same way as for the conjunction operation (lines 4-13). However, in this
case, we do not need to treat composite states that contain a final state
of one of the two automata separately. This follows from the semantics
of the disjunction operation, which accepts a word as soon as at least
one automaton accepts the word.

\begin{algorithm}
\caption{$\varrho_\orltl(\FA_1, \FA_2)$}
\label{alg:disjunction}
\DontPrintSemicolon
\KwIn{$\FA_1=(S_{\FA_1}, s_{01}, \Sigma, \delta_1, \{s_{f1}\}, T_{\FA_1})$ -- left DFA}
\KwIn{$\FA_2=(S_{\FA_2}, s_{02}, \Sigma, \delta_2, \{s_{f2}\}, T_{\FA_2})$ -- right DFA}
%\KwIn{{\it inf} -- boolean specifying DFAs are annotated}
\KwOut{DFA corresponding to $\CA{L}(\FA_1) \cup \CA{L}(\FA_2)$}
\BlankLine

$S'_{\FA_1} \asgn S_{\FA_1} \cup \{\trap\}$, $S'_{\FA_2} \asgn S_{\FA_2} \cup \{\trap\}$\;
$\delta'_1 \asgn \delta_1 \cup \{ (s, \sigma, \trap) \ |\ s\in S'_{\FA_1}\setminus \{s_{f1}\},\sigma\in\Sigma, \nexists \delta_1(s,\sigma) \}$\;
$\delta'_2 \asgn \delta_2 \cup \{ (s, \sigma, \trap) \ |\ s\in S'_{\FA_2}\setminus \{s_{f2}\},\sigma\in\Sigma, \nexists \delta_2(s,\sigma) \}$\;

$S \asgn \{(s_{01}, s_{02})\}$, $E \asgn \emptyset$\;
$stack \asgn [(s_{01}, s_{02})]$\;
\While{$stack \neq [\ ]$}{
    $\BS{s}=(s_1, s_2) \asgn stack.pop()$\;
    $S_n \asgn \{((s'_1, s'_2), \sigma) \ |\ \exists\sigma\in\Sigma \text{ s.t. }$\;
    \qquad\qquad\qquad $(s'_1 = \delta'_1(s_1, \sigma)) \wedge (s'_2 = \delta'_2(s_2, \sigma))\}$\;
    $E \asgn E \cup \{(\BS{s}, \sigma, \BS{s'}) \ |\ (\BS{s'}, \sigma) \in S_n \}$\;
    $S' \asgn \{\BS{s'} \ |\ \exists\sigma \in \Sigma \text{ s.t. } (\BS{s'}, \sigma) \in S_n \}$\;
    $stack.extends(S' \setminus S)$\;
    $S \asgn S \cup S'$\;
}
%$F \asgn \{ (s_1, s_2) \in S\ |\ (s_1=s_{f1}) \vee (s_2=s_{f2}) \}$\;
%\If{$\card{F} > 1$}{
%    $B \asgn \{ (\BS{s}, \sigma) \ |\  \exists \sigma \text{ s.t. }  (\BS{s}, \sigma, (s_{f1}, s_{f2})) \in E \}$\;
%    $L \asgn \{ (\BS{s}, \sigma) \ |\  \exists s_2\neq s_{f2}, \exists\sigma \text{ s.t. } (\BS{s}, \sigma, (s_{f1}, s_2) \in E \}$\;
%    $R \asgn \{ (\BS{s}, \sigma) \ |\  \exists s_1\neq s_{f1}, \exists\sigma \text{ s.t. } (\BS{s}, \sigma, (s_1, s_{f2}) \in E \}$\;
%    $choices \asgn \{B, L, R\}$\;
%    $E \asgn E \setminus \{(\BS{s}, \sigma, \BS{s'}) \in E\ |\ \BS{s'} \in F  \}$\;
%    $E \asgn E \cup \{(\BS{s}, \sigma, (s_{f1}, s_{f2})) \ |\ (\BS{s}, \sigma) \in B\cup L\cup R \}$
%}
$B \asgn \{ (\BS{s}, \sigma) \ |\  \exists \sigma \text{ s.t. }  (\BS{s}, \sigma, (s_{f1}, s_{f2})) \in E \}$\;
$L \asgn \{ (\BS{s}, \sigma) \ |\  \exists s_2\neq s_{f2}, \exists\sigma \text{ s.t. } (\BS{s}, \sigma, (s_{f1}, s_2) \in E \}$\;
$R \asgn \{ (\BS{s}, \sigma) \ |\  \exists s_1\neq s_{f1}, \exists\sigma \text{ s.t. } (\BS{s}, \sigma, (s_1, s_{f2}) \in E \}$\;
$F \asgn \{ (s_1, s_2) \in S\ |\ (s_1=s_{f1}) \vee (s_2=s_{f2}) \}$\;
$S \asgn S \setminus (F \cup \{(\trap, \trap)\})$\;
$E \asgn E \setminus \{(\BS{s}, \sigma, \BS{s'}) \in E\ |\ \BS{s'} \in F \}$\;
$E \asgn E \cup \{(\BS{s}, \sigma, (s_{f1}, s_{f2})) \ |\ (\BS{s}, \sigma) \in B\cup L\cup R \}$\;

$m_L = \big\{(u, \{ (u, v) \in S_\FA \}) \ |\ u\in S_{\FA_1} \big\}$\;
$m_R = \big\{(v, \{ (u, v) \in S_\FA \}) \ |\ v\in S_{\FA_2} \big\}$\;
$T_\FA \asgn tree(\orltl, relabelTree(T_{\FA_1}, m_L, \top),$\;
    \quad $relabelTree(T_{\FA_2}, m_R, \top), \{(s_{01}, s_{02})\}, \{(s_{f1}, s_{f2})\})$\;
$T_\FA.choices \asgn (B, L, R)$\;
$\FA \asgn (S, (s_{01}, s_{02}), \Sigma, E, \{(s_{f1}, s_{f2})\}, T_\FA)$\;
\Return $relabel(\FA, \emptyset, 0)$
\end{algorithm}

In the standard construction~\cite{Hopcroft2006}, the resulting automaton
would have multiple final states, which are computed in line 17.
However, because finals states do not have outgoing transitions,
we can merge all final states and obtain an automaton with only one final
state (lines 17-20). The composite trap state is also removed from the set
of states (line 18).

The annotation tree is created similarly to the conjunction case (lines 21-24).
However, for the disjunction case, we add additional information on the automaton.
This information $T.choices$ is used in latter algorithm to determine if a
word has satisfied the left, right, or both sub-formulae corresponding to the
disjunction formula. This is done by partitioning the transitions incoming
into finals states (line 14-16) and storing this partition in the associated
tree node (line 25). Note that only the start state and the symbol of each
transition is stored in the partition sets and these are well defined,
because the DFAs are deterministic.

%\begin{algorithm}
%\caption{$annotate_\times(op, \FA, \FA_1, \FA_2)$}
%\label{alg:annotate-product}
%\DontPrintSemicolon
%\KwIn{$op$ -- operation: either $\orltl$ or $\andltl$}
%\KwIn{$\FA=(S_\FA, s_0, \Sigma, \delta, \{s_f\}, T_\FA)$ -- product DFA}
%\KwIn{$\FA_1=(S_{\FA_1}, s_{01}, \Sigma, \delta_1, \{s_{f1}\}, T_{\FA_1})$ -- left DFA}
%\KwIn{$\FA_2=(S_{\FA_2}, s_{02}, \Sigma, \delta_2, \{s_{f2}\}, T_{\FA_2})$ -- right DFA}
%\KwOut{Tree corresponding to product DFA $\FA$}
%\BlankLine
%
%$m_L = \big\{(u, \{ (u, v) \in S_\FA \}) \ |\ u\in S_{\FA_1} \big\}$\;
%$m_R = \big\{(v, \{ (u, v) \in S_\FA \}) \ |\ v\in S_{\FA_2} \big\}$\;
%\Return $tree(op, relabelTree(T_{\FA_1}, m_L, \top),$\;
%           \qquad\qquad\qquad $relabelTree(T_{\FA_2}, m_R, \top), \{s_0\}, \{s_f\})$\;
%
%%    # update init and final states of AST nodes
%%    if getOptimizationFlag():
%%        dfa_dest.tree.normalize(init=dfa_dest.init.keys(), final=dfa_dest.final)
%\end{algorithm}

\subsubsection{Concatenation}
\label{sec:concatenation}

The algorithm to compute an automaton accepting the concatenation language
of two languages is shown in Alg.~\ref{alg:concatenation}. The special structure of
the unambiguous languages, see Sec.~\ref{sec:props} for details, admits a particularly
simple and intuitive construction procedure. The composite automaton is obtained
by identifying the final state of left automaton $\FA_1$ with the initial state of the
right automaton $\FA_2$.

\begin{algorithm}
\caption{$\varrho_\cdot(\FA_1, \FA_2)$}
\label{alg:concatenation}
\DontPrintSemicolon
\KwIn{$\FA_1=(S_{\FA_1}, s_{01}, \Sigma, \delta_1, \{s_{f1}\}, T_{\FA_1})$ -- left DFA}
\KwIn{$\FA_2=(S_{\FA_2}, s_{02}, \Sigma, \delta_2, \{s_{f2}\}, T_{\FA_2})$ -- right DFA}
%\KwIn{{\it inf} -- boolean specifying DFAs are annotated}
\KwOut{DFA corresponding to $\CA{L}(\FA_1) \cdot \CA{L}(\FA_2)$}
\BlankLine

$\FA_1 \asgn relabel(\FA_1, \emptyset, 0)$\;
$\FA_2 \asgn relabel(\FA_2, \{(s_{02}, s_{f1})\}, \card{S_{\FA_1}})$\;
\lIf{{\it inf}}{
    $T \asgn tree(\cdot, T_{\FA_1}, T_{\FA_2}, \{s_{01}\}, \{s_{f2}\})$\;
}\lElse{
    $T \asgn \emptyset$\;
}
\Return $(S_{\FA_1} \cup S_{\FA_2}, s_{01}, \Sigma, \delta_1 \cup \delta_1, \{s_{f2}\}, T)$\;
\end{algorithm}

\subsubsection{Within}

There are two algorithms used to construct a DFA associated with a {\em within} operator,
Alg.~\ref{alg:within-inf} and Alg.~\ref{alg:within-normal} correspond to the relaxed
and normal construction (lines 6-9 of Alg.~\ref{alg:twtl2dfa}).

{\it Relaxed within:} %We start with algorithm for the relaxed case, because it is simpler.
The construction procedure Alg.~\ref{alg:within-inf} is as follows: starting from the
DFA corresponding to the enclosed formula, all states are connected via blocking symbols
to the initial state (lines 3-4).
The last step is to create a number of $a$ states connected in sequence
for all symbols, similarly to Alg.~\ref{alg:hold}, and connecting the $a$-th state to the initial
state also for all symbols (lines 5-8).

\begin{algorithm}
\caption{$\varrho_\infty(\FA, a, b)$}
\label{alg:within-inf}
\DontPrintSemicolon
\KwIn{$\FA=(S_\FA, s_0, \Sigma, \delta, \{s_f\}, T_\FA)$ -- child DFA}
\KwIn{$a$ -- lower bound of time-window}
\KwIn{$b$ -- upper bound of time-window}
\KwOut{computed DFA}
\BlankLine

$\FA \asgn relabel(\FA, \emptyset, 0)$\;
$S \asgn \emptyset$, $E \asgn \emptyset$\;
\For{$s \in S_\FA \setminus \{s_f\}$}{
    $E \asgn E \cup \{(s, \sigma, s_0) \ |\ \nexists s'=\delta(s, \sigma)\}$
}
\If{$a>0$}{
    $S \asgn \{\card{S_\FA}, \ldots,\card{S_\FA} + a-1 \}$\;
    $E \asgn E \cup \big\{(i, \sigma, i+1) \ |\ i\in S\setminus \{\card{S_\FA}+ a-1\}, \sigma \in \Sigma \big\}$\;
    $E \asgn E \cup \{(\card{S_\FA}+ a-1, \sigma, s_0) \ |\ \sigma \in \Sigma \}$\;
}
$T \asgn tree([\,]^{[a, b]}_\infty, T_\FA, \emptyset, \{\card{S_\FA}\}, \{s_f\})$\;
\Return $(S_\FA \cup S, \card{S_\FA}, \Sigma, \delta \cup E, \{s_f\}, T)$
\end{algorithm}

Connecting all states to the initial state represents a restart of the automaton in case a
blocking symbol was encountered. Thus, the resulting automaton offers infinite retries
for a word to satisfy the enclosed formula.
The $a$ states added before the initial state represent a delay of length $a$ for the
start of the tracking of the satisfaction of the enclosed formula.
Note that the procedure and resulting automaton do not depend on the upper bound $b$.

{\it Normal within:} The algorithm for the normal case builds upon Alg.~\ref{alg:within-inf}.
In this case
the construction procedure Alg.~\ref{alg:within-normal} must take into account the upper
time bound $b$. Similarly to the relaxed case, we need to restart the automaton of the
when a blocking symbol is encountered. However, there are two major differences:
(i) the automaton must track the number of restarts, because there are only a finite
number of tries depending on the deadline $b$, and
(ii) the automaton $\FA$ may need to be truncated for the last restart retries,
i.e., all paths must have a length of at most a given length, in order to ensure that
the satisfaction is realized before the upper time limit $b$.

\begin{algorithm}
\caption{$\varrho_{[\,]}(\FA, a, b)$}
\label{alg:within-normal}
\DontPrintSemicolon
\KwIn{$\FA=(S_\FA, s_0, \Sigma, \delta, \{s_f\}, T_\FA)$ -- child DFA}
\KwIn{$a$ -- lower bound of time-window}
\KwIn{$b$ -- upper bound of time-window}
\KwOut{computed DFA}
\BlankLine

%compute the maximum number of restarts\;
$l \asgn Dijkstra(\FA, s_0, s_f)$\;
$p \asgn b - a - l + 2$\;
$I \asgn [\,]$ \tcp*[f]{list}\;
$n \asgn 0$\;
$\FA_r \asgn (S_{\FA_r}=\emptyset, \infty, \Sigma, \delta_r=\emptyset, \emptyset, \emptyset)$\;
\For{$k \in \{1,\ldots, p\}$}{
    $m \asgn \{(s_f, -1)\}$  \tcp*[f]{mark final state}\;
    $\FA_a \asgn relabel(\FA, m, n)$\;
    $\FA_t \asgn truncate(\FA_a, b-a+2-k)$\;
    $\FA_r \asgn (S_{\FA_r} \cup S_{\FA_t}, \infty, \Sigma, \delta_r \cup \delta_t, \{-1\},  \emptyset)$\;
    $I \asgn I + [s_{0t}]$\;
    $n \asgn n + \card{S_{\FA_t}}$\;
}
$S_c \asgn \{I[0]\}$, $E \asgn \emptyset$\;
\For{$s_r \in I[1:\,]$}{
    $S_n \asgn \emptyset$\;
    \For{$s \in S_c \setminus \{-1\}$}{
        $E \asgn E \cup \{ (s,\sigma,s_r) \ |\ \sigma\in\Sigma \text{ s.t. } \nexists \delta_r(s,\sigma)  \}$\;
    }
    $S_c \asgn S_c \cup \{s_r\}$\;
}
$S \asgn \emptyset$\;
\If{$a>0$}{
    $S \asgn \{\card{S_{\FA_r}}, \ldots,\card{S_{\FA_r}} + a-1 \}$\;
    $E \asgn E \cup \big\{(i, \sigma, i+1) \ |\ i\in S\setminus \{\card{S_{\FA_r}}+ a-1\}, \sigma \in \Sigma \big\}$\;
    $E \asgn E \cup \{(\card{S_{\FA_r}}+ a-1, \sigma, s_0) \ |\ \sigma \in \Sigma \}$\;
}
\Return $(S_{\FA_r} \cup S, I[0], \Sigma, \delta_r \cup E, \{-1\}, \emptyset)$
\end{algorithm}

In Alg.~\ref{alg:within-normal}, first the maximum number of restarts $p$ is computed
in lines 1-2.  Then, $p$ DFAs are created (lines 3-12), which correspond to the relabeled
and truncated copies of $\FA$, see Alg.~\ref{alg:truncate}, and their union is computed
iteratively. The truncation bound is computed as the remaining time units until the limit
$b$ is reached. The final state is always labeled with $-1$ (line 7) and, therefore, the
resulting DFA has exactly one final state.
Next, the restart transitions are added (lines 13-18). Note that the transitions, enabled
by blocking symbols, lead to initial states of the proper restart automaton. For example,
if a blocking symbol was encountered after two symbols, then the restart transition
(if it exists) leads to the initial state of the fourth copy of the automaton.
Lastly, a delay of $a$ time units is added before the initial
state of the automaton similar to the relaxed case.

\subsubsection{Truncate}

Alg.~\ref{alg:truncate} takes as input a DFA $\FA$ and a cutoff bound
$l$ and returns a version of $\FA$ with all paths guaranteed to have length at most $l$.
The algorithm is based on a breath-first search and returns a strict DFA.

\begin{algorithm}
\caption{$truncate(\FA, l)$}
\label{alg:truncate}
\DontPrintSemicolon
\KwIn{$\FA=(S_\FA, s_0, \Sigma, \delta, \{s_f\}, T_\FA)$ -- a DFA}
\KwIn{$l$ -- cutoff value}
\KwOut{computed DFA}
\BlankLine

$S \asgn \{s_0\}$\;
$E \asgn \emptyset$\;
$L_n \asgn \{s_0\}$\;
\For{$i\in \{1,\ldots,l\}$}{
    $L_c \asgn L_n$\;
    $L_n \asgn \emptyset$\;
    \For{$s \in L_c$}{
        \For{$(s_c,\sigma_c) \in \{(s',\sigma) | \exists \sigma\in\Sigma \text{ s.t. } s\ras{\sigma}_\FA s' \}$}{
                $E \asgn E \cup (s, \sigma_c, s_c))$\;
                \If{$s_c \notin S$}{
                    $S \asgn S \cup \{s_c\}$\;
                    $L_n \asgn L_n \cup \{s_c\}$
                }
        }
    }
}
$\FA_t = (S_\FA, s_0, \Sigma, \delta\setminus E, \{s_f\}, T_\FA)$\;
$S_{traps} = \{s \in S_\FA | \nexists \BS{\sigma}\in\Sigma^* \text{ s.t. } s\ras{\BS{\sigma}}_{\FA_t} s_f\}$\;
\Return $(S_\FA \setminus S_{traps}, s_0, \Sigma, \delta\setminus E, \{s_f\}, T_\FA)$
\end{algorithm}

\subsection{Correctness}
\label{sec:correctness}

%In this section we prove that the proposed algorithm for translating
%TWTL formulae to automata is correct.
The following theorems show that the proposed algorithms
for translating TWTL formulae to (normal or annotated)
automata are correct.

\begin{theorem}
\label{th:twtl2dfa-correctness-infinity}
If $\phi$ is a TWTL formula satisfying Assump.~\ref{assump:spec}
and the global parameter {\it inf} is true,
then Alg.~\ref{alg:twtl2dfa} generates a DFA $\FA_\infty$
such that $\CA{L}(\FA_\infty) = \CA{L}(\phi(\infty))$.
\end{theorem}
\begin{proof}
The proof follows by structural induction on $AST(\phi)$ and the
properties of TWTL languages.

Before we proceed with the induction, notice that all 
construction algorithms associated with the operators of TWTL
generate strict DFAs with only one final state without any
outgoing transitions.

The base case corresponds to the leaf nodes of $AST(\phi)$
which are associated with {\em hold} operators, see
Fig.~\ref{fig:ast-example}, and follows by construction from
Alg.~\ref{alg:hold}.

%For the induction case, we assume that the theorem holds
%for the DFAs returned by the induction in Alg.~\ref{alg:twtl2dfa}.
The induction hypothesis requires that the theorem holds
for the DFAs returned by the recursion in Alg.~\ref{alg:twtl2dfa}.
In the case of the conjunction and disjunction operators,
the property follows from the product construction
method~\cite{Hopcroft2006}.
The theorem holds also for the concatenation operator, because:
(a) the returned DFAs have one final state without any
outgoing transitions, and
(b) the languages corresponding to the two operand formulae
are unambiguous.
Thus, the correctness of the construction described in
Alg.~\ref{alg:concatenation} follows immediately from the
unambiguity of the concatenation, see
Def.~\ref{def:unambiguous-concat}.
Lastly, the case of the {\em within} operator (relaxed form),
follow from the Assump.~\ref{assump:spec}.
The {\em within} operator adds transitions to a DFA from each state
to the initial state on all undefined symbols.
In other words, the operator restarts the execution of a DFA
from the initial state.
If there are no disjunction operators, then going back to
the initial state is the only correct choice.
Otherwise, because of alternative paths induced by
disjunction, there might be other states from which the DFA
might need to go back to in order to correctly restart.
\end{proof}

\begin{theorem}
\label{th:twtl2dfa-correctness}
If $\phi$ is a TWTL formula satisfying Assump.~\ref{assump:spec}
and the global parameter {\it inf} is false,
then Alg.~\ref{alg:twtl2dfa} generates DFA $\FA$
such that $\CA{L}(\FA) = \CA{L}(\phi)$.
\end{theorem}
\begin{proof}
The proof is similar to that of Thm.~\ref{th:twtl2dfa-correctness-infinity}
and is omitted for brevity.
\end{proof}

\subsection{Complexity}
\label{sec:complexity-automata-construction}

In this section, we review the complexity of the algorithms presented
in the previous section for the construction of DFAs from
TWTL formulae. The complexity of basic composition operations
for incomplete and acyclic DFAs has been explored
in~\cite{Maia13,Han07,Campeanu01,Gao11,Daciuk2003}.
Our construction algorithms differ from the ones in the literature
because we specialized and optimized them to translate TWTL formulae
and handle words over power sets of atomic propositions.

The complexity of the relabeling procedures are $O(\card{T})$ and
$O(\card{S_\FA})$ corresponding to Alg.~\ref{alg:relabel-tree} and
Alg.~\ref{alg:relabel}, respectively.
The complexity of the {\em hold} operator Alg.~\ref{alg:hold} is
$O(d \cdot \spow{\card{\AP}})$.
The construction algorithms for conjunction and disjunction
Alg.~\ref{alg:conjunction} and Alg.~\ref{alg:disjunction} have the
same complexity $O(\card{S_{\FA_1}} \cdot \card{S_{\FA_2}} \cdot \spow{\card{\AP}})$,
because these are based on the product automaton construction.
Concatenation has complexity $O(\card{S_{\FA_1}} + \card{S_{\FA_2}})$
due to the relabeling operations.
Lastly, the {\em within} operation can be performed
in $O(a\cdot\spow{\card{\AP}} + \card{S_\FA} \cdot\spow{\card{\AP}})$
and $O(a\cdot\spow{\card{\AP}} + b\card{S_\FA} \cdot\spow{\card{\AP}})$
for the infinity Alg.~\ref{alg:within-inf} and the normal Alg.~\ref{alg:within-normal}
construction, respectively, where Alg.~\ref{alg:truncate} used in the normal
construction procedure takes $O(\card{S_\FA} \cdot\spow{\card{\AP}})$.
The overall translation algorithm Alg.~\ref{alg:twtl2dfa}
takes at most $O(\spow{\card{\phi} + \card{\AP}})$.

It is very important to notice that the infinity construction does not
depend on the deadline $b$, which makes the procedure more
efficient than the normal construction.

\section{Solution}
\label{sec:algorithms}

In this section, we will use the following notation.
Let $T$ be an annotation tree associated with a DFA.
We denote by $\phi_T$ the TWTL formula corresponding
to the tree $T$.
Given a finite sequence $\BF{p} = p_0, \ldots, p_n$,
we denote the first and the last elements by $b(\BF{p}) = p_0$
and $e(\BF{p}) = p_n$, respectively.

\begin{definition}[Primitive]
Let $\phi$ be a TWTL formula. We say that $\phi$ is
{\em primitive} if $\phi$ does not contain any {\em within} operators.
\end{definition}

\subsection{Compute temporal relaxation for a word}

%The most basic operation we can perform using
%the construction presented in the previous section is
%to compute the temporal relaxation of a word $\BS{\sigma}$
%with respect to a TWTL formula $\phi$.

The automata construction presented in
Sec.~\ref{sec:construction} can be used to compute the
temporal relaxation of words with repsect to TWTL
formulae.
Let $\phi$ be a TWTL formula and $\BS{\sigma}$
be a word.
In this section, we show how to infer (synthesize) a set of
temporal relaxations $\BS{\tau}$ of the deadlines in $\phi$
such that $\BS{\sigma}$ satisfies $\phi(\tau)$ and
$\trel{\BS{\tau}}$ is minimized.
Alg.~\ref{alg:compute-temporal-relaxation} computes
the vector of temporal relaxations corresponding to each
{\em within} operator.
First, the annotated DFA $\FA_\infty$ is computed together
with the associated annotation tree $T$ (line 2).
Next, additional annotations are added to the tree $T$ using
the $initTreeTR()$ procedure (line 3).
Each node corresponding to a {\em within} operation is
assigned three variables $T.ongoing$, $T.done$ and $T.steps$,
which track whether the processing of the operator is
ongoing, done, and the number of steps to process the operator,
respectively.
The three variables are initialized to $\perp$, $\perp$, and
$-1$, respectively.
Then, Alg.~\ref{alg:compute-temporal-relaxation} cycles
through the symbols of the input word $\BS{\sigma}$ and
updates the tree using $updateTree()$ via Alg.~\ref{alg:tr-update-tree}.
Finally, the temporal relaxation vector is returned by the
$evalTreeTR()$ procedure via Alg.~\ref{alg:tr-eval-tree}.

\begin{algorithm}
\caption{$tr(\cdot)$ -- Compute temporal relaxation}
\label{alg:compute-temporal-relaxation}
\DontPrintSemicolon
\KwIn{$\BS{\sigma}$ a word over the alphabet $\spow{\AP}$}
\KwIn{$\phi$ a TWTL formula}
\KwOut{$\tau^*$ - minimum maximal temporal relaxation}
\KwOut{$\BS{\tau}$ - temporal relaxation vector}
\BlankLine

\lIf{$\phi$ is primitive}{
    \Return{$(-\infty, [\ ])$}\;
}
$\FA_\infty, T \asgn translate(\phi; \inf=\True)$\;
$initTreeTR(T)$\;
$s_{prev} \asgn \perp$;
$s_{c} \asgn s_0$\;
$updateTreeTR(T, s_c, s_{prev}, \emptyset, \emptyset)$\;
\For{$\sigma \in \BS{\sigma}$}{
    \lIf{$s_c \in F_{\FA_\infty}$}{
        \Break\;
    }
    $s_{prev} \asgn s_c$\;
    $s_c \asgn \delta_{\FA_\infty}(s_c, \sigma)$\;
    $updateTreeTR(T, s_c, s_{prev}, \sigma, \emptyset)$\;
}
\Return{$evalTreeTR(T)$}
\end{algorithm}

The tree is updated recursively in Alg.~\ref{alg:tr-update-tree}.
A {\em within} operator is marked as ongoing, i.e.,
$T.ongoing = \True$, when the current state is in the set of
initial states associated with the operator (line 2).
Similarly, when the current state is in the set of final states
associated with the operator, the {\em within} operator is
marked as done (lines 3-6), i.e. $T.done = \True$ and
$T.ongoing = \False$.
The number of steps $T.steps$ of all ongoing {\em within}
operators is incremented (line 7).

To enforce correct computation of the temporal relaxation
with respect to the disjunction operators, Alg.~\ref{alg:tr-update-tree}
keeps track of a set of constraints $C$.
The set $C$ is composed of state-symbol pairs, and
is used to determine which of the two subformulae of a
disjunction are satisfied by the input word (lines 12-17).
To achieve this, we use the annotation variables $T.choices$
(see Alg.~\ref{alg:disjunction}), which capture both cases.
For all other operators, the constraint sets are propagated
unchanged (lines 8, 10, 11).

\begin{algorithm}
\caption{$updateTreeTR(\cdot)$}
\label{alg:tr-update-tree}
\DontPrintSemicolon
\KwIn{$s_c$ -- current state}
\KwIn{$s_{prev}$ -- previous state}
\KwIn{$\sigma$ -- current symbol in word}
\KwIn{$C$ -- set of constraints associated with the states}
\BlankLine
\eIf{$T.op = [\ ]^{[a, b]}$}{
    \lIf{$s_c \in T.I$}{
         $T.ongoing \asgn \True$\;
    }
    \If{$s_c \in T.F$}{
         \If{$(C = \emptyset) \orltl (\sigma \subseteq C(s_{prev}))$}{
              $T.ongoing \asgn \perp$\;
              $T.done \asgn \True$
         }
    }
    \lIf{$T.ongoing$}{
        $T.\tau \asgn T.\tau + 1$\;
    }
    $updateTreeTR(T.left, s_c, s_{prev}, \sigma, C)$
}{
    \lIf{$T.op = \cdot$}{
        $C_L \asgn \emptyset$; $C_R \asgn C$\;
    }
    \lElseIf{$T.op = \andltl$}{
        $C_L \asgn C$; $C_R \asgn C$\;
    }
    \ElseIf{$T.op = \orltl$}{
         $C_L \asgn T.choices.L \cup T.choices.B$\;
         $C_R \asgn T.choices.R \cup T.choices.B$\;
         \If{$C\neq \emptyset$}{
             $C_L \asgn C \cap C_L$\;
             $C_R \asgn C \cap C_R$
         }
    }
    $updateTreeTR(T.left, s_c, s_{prev}, \sigma, C_L)$\;
    $updateTreeTR(T.right, s_c, s_{prev}, \sigma, C_R)$\;
}
\end{algorithm}

Finally, Alg.~\ref{alg:tr-eval-tree} extracts the temporal
relaxation from the annotation tree $T$ after all symbols
of the input word $\BS{\sigma}$ were processed.
Alg.~\ref{alg:tr-eval-tree} also computes the minimum
maximum temporal relaxation value, which may be
$-\infty$ if $\phi$ is primitive (line 1).
The recursion in Alg.~\ref{alg:tr-eval-tree} differs
between disjunction and the other operators.
One subformula is sufficient to hold to satisfy the formula
associated with a disjunction operator.
Thus, the optimal temporal relaxation is the
minimum or maximum between the two optimal
temporal relaxations of the subformulae for disjunction
(line 12), and conjunction and concatenation (line 13),
respectively.
Lines 15-16 of Alg.~\ref{alg:tr-eval-tree} cover the
cases involving primitive subformulae.

The complexity of Alg.~\ref{alg:compute-temporal-relaxation}
is $O(\spow{\card{\phi}+\card{\AP}} +  \card{\BS{\sigma}} \cdot \card{\phi})$,
where the first term is the complexity of constructing $\FA_\infty$
in line 1 and the second term corresponds to the update of the
tree for each symbol in $\BS{\sigma}$ and the final evaluation of
the tree.

\begin{algorithm}
\caption{$evalTreeTR(\cdot)$}
\label{alg:tr-eval-tree}
\DontPrintSemicolon
\KwIn{$T$ -- annotated tree}
\KwOut{$\tau^*$ - minimum maximal temporal relaxation}
\KwOut{$\BS{\tau}$ - temporal relaxation vector}
\BlankLine
\lIf{$\phi_T$ is primitive}{
    \Return $(-\infty, [\ ])$\;
}
\uElseIf{$T.op = [\phi]^{[a,b]}$}{
    $\tau^*_{ch}, \BS{\tau}_{ch} = evalTreeTR(tree.left)$\;
    \eIf{$T.done = \True$}{
        \Return {\small $\big(\max\{\tau^*_{ch}, T.steps - b\},  [\BS{\tau}_{ch},  T.steps - b]\big)$ }
    }{
        \Return $\big(-\infty, [\BS{\tau}_{ch}, -\infty]\big)$
    }
}
\Else(\tcp*[h]{$\andltl$, $\orltl$ or $\cdot$}){
    $\tau^*_L, \BS{\tau}_L = evalTreeTR(tree.left)$\;
    $\tau^*_R, \BS{\tau}_R = evalTreeTR(tree.right)$\;
    \eIf{$(\tau^*_L \neq -\infty) \andltl (\tau^*_R \neq -\infty)$}{
        \lIf{$T.op = \orltl$}{
            $\tau^* \asgn \min\{\tau^*_L, \tau^*_R\}$\;
        }\lElse{
            $\tau^* \asgn \max\{\tau^*_L, \tau^*_R\}$\;
        }
    }{
       \lIf{$T.op = \orltl$}{
           $\tau^* \asgn \max\{\tau^*_L, \tau^*_R\}$\;
       }\lElse{
           $\tau^* \asgn -\infty$\;
       }
    }
    \Return $\big(\tau^*, [\BS{\tau}_L, \BS{\tau}_R]\big)$
}
\end{algorithm}

\subsection{Control policy synthesis for a finite transition system}

Let $\TS$ be a finite transition system, and $\phi$
a specification given as a TWTL formula.
The procedure to synthesize an optimal control policy
by minimizing the temporal relaxation has three steps:
\begin{enumerate}
    \item constructing the annotated DFA $\FA_\infty$ corresponding to $\phi$,
    \item constructing the synchronous product $\PA =\TS \times \FA_\infty$
    between the transition system $\TS$ and the annotated
    DFA $\FA_\infty$,
    \item computing the optimal policy on $\PA$ using
    Alg.~\ref{alg:policy} and generating the optimal trajectory of $\TS$
    from the optimal trajectory of $\PA$ by projection,
%    \item the annotated DFA $\FA_\infty$ corresponding
%    to $\phi$ is constructed;
%    \item the synchronous product $\PA =\TS \times \FA_\infty$
%    between the transition system $\TS$ and the annotated
%    DFA $\FA_\infty$ is constructed;
%    \item the policy is computed on $\PA$ using
%    Alg.~\ref{alg:policy}. The optimal trajectory of $\TS$ is
%    obtained from the optimal trajectory of $\PA$ by projection.
\end{enumerate}
where the synchronous product $\PA$ is defined as follows:

\begin{definition}[Product Automaton]
\label{def:pa}
Given a TS $\TS = (X, x_0, \Delta, \AP, h)$ and
a DFA $\FA = (S_\FA, s_0, \spow{\AP}, \delta_\FA, F_\FA)$,
their product automaton, denoted by $\PA = \TS \times \FA$,
is a tuple $\PA = (S_\PA, p_0, \Delta_\PA, F_\PA)$ where:
\begin{itemize}
    \item $p_0 = (x_0, s_0)$ is the initial state;
    \item $S_\PA \subseteq X \times S_\FA $ is a finite set of states that are reachable from the initial state: for every $(x^*, s^*) \in S_\PA$, there exists a sequence of $\BF{x} = x_0 x_1 \ldots x_n x^*$, with $x_k \ra_\TS x_{k+1}$ for all $0 \leq k < n$ and $x_n \ra_\TS x^*$, and a sequence $\BF{s} = s_0 s_1 \ldots s_n s^*$ such that $s_0$ is the initial state of $\FA$, $s_k \ras{h(x_{k+1})}_\BA s_{k+1}$ for all $0 \leq k < n$ and $s_n \ras{h(x^*)}_\TS s^*$;
    \item $\Delta_\PA \subseteq S_\PA \times S_\PA$ is the set of transitions defined by: $((x, s), (x', s')) \in \Delta_\PA$ iff $x \ra_\TS x'$ and $s \ras{h(x')}_\BA s'$;
    \item $F_\PA = (X \times F_\FA) \cap S_\PA$ is the set of accepting states of $\PA$.
\end{itemize}
\end{definition}

A transition in $\PA$ is also denoted by
$(x, s) \ra_\PA (x', s')$ if $((x, s), (x', s')) \in \Delta_\PA$.
A trajectory $\BF{p} = (x_0, s_0) (x_1, s_1) \ldots$ of $\PA$
is an infinite sequence, where $(x_0, s_0) = p_0$ and
$(x_k, s_k) \ra_\PA (x_{k+1}, s_{k+1})$ for all $k \geq 0$.
A trajectory of $\PA = \TS \times \FA$ is said to be accepting
if and only if it ends in a state that belongs to the set of
final states $F_\PA$.
It follows by construction that a trajectory
$\BF{p} = (x_0, s_0) (x_1, s_1) \ldots$ of $\PA$
is accepting if and only if the trajectory $s_0 s_1 \ldots$
is accepting in $\FA$.
As a result, a trajectory of $\TS$ obtained from an accepting
trajectory of $\PA$ satisfies the given specification encoded by $\FA$.
%For $x \in X$, we define $\beta_\PA(x) = \{ s \in S_\BA : (x, s) \in S_\PA \}$ as the set of \buchi automaton states
%that correspond to $x$ in $\PA$.
We denote the projection of a trajectory
$\BF{p} = (x_0, s_0) (x_1, s_1) \ldots$ onto
$\TS$ by $\gamma_\TS(\BF{p}) = x_0 x_1 \ldots$.

Before we present the details of the proposed algorithm,
we want to point out that completeness may be decided
easily by using the product automaton $\PA$.
That is, testing if there exists a temporal relaxation such that
a satisfying policy in $\TS$ may be synthesized can
be performed very efficiently as shown by the following
theorem.

\begin{theorem}
\label{th:tr-completeness}
Let $\phi$ be a TWTL formula and $\TS$ be a finite transition system.
Deciding if there exists a finite $\BS{\tau} \in \BB{Z}^m$ and a trajectory $\BF{x}$
of $\TS$ such that $\BF{o} \models \phi(\BS{\tau})$, can be performed
in $O(\card{\Delta} \cdot \card{\delta_{\FA_\infty}})$,
where $m$ is the number of {\em within} operators in $\phi$,
$\FA_\infty$ is the annotated DFA corresponding to $\phi$,
$\BF{o}$ is the output trajectory induced by $\BF{x}$,
and $\Delta$ and $\delta_{\FA_\infty}$ are the sets
of transitions of $\TS$ and $\FA_\infty$, respectively.
\end{theorem}
\begin{remark}
The complexity in Thm.~\ref{th:tr-completeness} is
independent of the deadlines of the {\em within} operators $\phi$.
\end{remark}
\begin{proof}
The result follows immediately using Dijkstra's algorithm on
the product automaton $\PA$.
\end{proof}

Note that Dijkstra's algorithm may not necessarily provide an optimal
trajectory of $\TS$ with respect to the minimum maximum
temporal relaxation of the induced output word.
Thus, we present a Dijkstra-based procedure to compute an optimal policy
using the product automaton $\PA$.
The proposed solution is presented in Alg.~\ref{alg:policy},
which describes a recursive procedure over an
annotated AST tree $T$.

\begin{algorithm}
\caption{Policy synthesis -- $policy(T, \PA)$}
\label{alg:policy}
\DontPrintSemicolon
\KwIn{$T$ -- the annotation AST tree}
\KwIn{$\PA$ -- product automaton}
\BlankLine
\If{$\phi_T$ is primitive}{
$M =\{\BF{p}\ |\ b(\BF{p}) \in T.I, e(\BF{p}) \in T.F\}$\;
$\tau^*[\BF{p}] = -\infty, \forall \BF{p}\in M$\;
}
\ElseIf{$T.op = [\,]^{[a, b]}$ $\wedge$ $\phi_{T.left}$ is primitive}{
$M =\{\BF{p}\ |\ b(\BF{p}) \in T.I, e(\BF{p}) \in T.F\}$\;
$\tau^*[\BF{p}] = \card{\BF{p}} - b, \forall \BF{p} \in M$\;
}
\ElseIf{$T.op = [\,]^{[a, b]}$ $\wedge$ $\phi_{T.left}$ is not primitive}{
$M_{ch}, \tau^{max}_{ch} = policy(T.left, \PA)$\;
$M = \{p_i\ras{a}p\ras{*}p' \ |\  p_i \in T.I, p\ras{*}p' \in M_{ch}\}$\;
$\tau^*[\BF{p}]=\max\{\card{\BF{p}} - b, \tau^*_{ch}[\BF{p}]\}, \forall \BF{p} \in M$\;
}
\ElseIf{$T.op = \cdot$}{
$M_L, \tau^*_L = policy(T.left, \PA)$\;
$M_R, \tau^*_R = policy(T.right, \PA)$\;
$M = \{ \BF{p}_1 \cdot \BF{p}_2 \ |\  \BF{p}_1 \in M_L, \BF{p}_2 \in M_R, e(\BF{p}_1) \ra_\PA b(\BF{p}_2)
\}$\;
$\tau^*[\BF{p}] = \max \{\tau^*_L(\BF{p}), \tau^*_R(\BF{p})\}$, $\forall \BF{p}\in M$\;
}
\ElseIf{$T.op = \orltl$}{
$M_L, \tau^*_L = policy(T.left, \PA)$\;
$M_R, \tau^*_R = policy(T.right, \PA)$\;
$M = M_L \cup M_R$\;
$\tau^*[\BF{p}] = \begin{cases}
\tau^*_L[\BF{p}] & \BF{p} \in M\setminus M_R\\
\tau^*_R[\BF{p}] & \BF{p} \in M\setminus M_L\\
\min\{\tau^*_L[\BF{p}], \tau^*_R[\BF{p}]\} & \BF{p} \in M_L\cap M_R
\end{cases}
$
}
\ElseIf{$T.op = \andltl$}{
$M_L, \tau^*_L = policy(T.left, \PA)$\;
$M_R, \tau^*_R = policy(T.right, \PA)$\;
$M = M_L \cap M_R$\;
$\tau^*[\BF{p}] = \max \{\tau^*_L(\BF{p}), \tau^*_R(\BF{p})\}$, $\forall \BF{p}\in M$\;
}

\Return ($M$, $\tau^*$)
\end{algorithm}

The recursive procedure in Alg.~\ref{alg:policy} has six cases.
The first case (lines 1-3) corresponds to a primitive formula.
In this case, there are no deadlines to relax since the formula
does not contain any {\em within} operators. Thus, solutions (if any
exist) can be computed using Dijkstra's algorithm.
The next two cases treat the {\em within} operators. In the former case
(lines 4-5), the enclosed formula is a primitive formula and
the only deadline which must be optimized is the one associated
with the current {\em within} operator. In the latter case (lines 7-10),
the enclosed formula is not primitive. Therefore, there are
multiple deadlines that must be considered. To optimize the
temporal relaxation $\trel{\cdot}$, we take the maximum
between the previous maximum temporal relaxation and the
current temporal relaxation (line 10).
The fourth case (lines 11-15) handles the concatenation operator.
First, the paths and the corresponding temporal relaxations are
computed for the left and the right subformulae in lines 12 and
13, respectively. Afterwards, the paths satisfying the left
subformula are concatenated to the paths satisfying the right
formula. However, the concatenation of paths $p_L$ and $p_R$
is restricted to pairs which have the following property:
there exists a transition in $\PA$ between the last state
of $p_L$ and the first state in $p_R$.
The temporal relaxation of the concatenation of two paths
is the maximum between the temporal relaxations of the
two paths (line 15).
The next case is associated with the disjunction operator
(lines 16-20). As in the concatenation case, first the
paths satisfying the left $M_L$ and the right $M_R$
subformulae are computed in lines 17 and 18, respectively.
The set corresponding to the disjunction of the two
formulae is the union of the two sets because the
paths must satisfy either one of the two subformulae.
The temporal relaxation of a path $p$ in the union is
computed as follows (line 20):
(a) if a path is only in the left, $\BF{p} \in M_L \setminus M_R$,
or only in the right set, $\BF{p} \in M_R \setminus M_L$, then
the temporal relaxation is $\tau^*_L[\BF{p}]$ or 
$\tau^*_R[\BF{p}]$, respectively;
(b) the path is in both sets, $\BF{p} \in M_L \cap M_R$,
then the temporal relaxations is the minimum of the two
previously computed ones,
$\min\{\tau^*_L[\BF{p}], \tau^*_R[\BF{p}] \}$.
In the case (a), $\BF{p}$ satisfies only one subformula
and, therefore, only one temporal relaxation is available.
In the case (b), $\BF{p}$ satisfies both subformulae.
Because only one is needed, the subformula that yields
the minimum temporal relaxation is chosen, i.e.,
the minimum between the two temporal relaxations.
The last case handles the conjunction operator (lines 21-25).
As in the previous two cases, the paths satisfying the
left and the right subformulae are computed first (lines 22-23).
Then the intersection of the two sets is computed as the
set of paths satisfying the conjunctions because the paths
must satisfy both subformulae.
The temporal relaxations of the paths in the intersections
are computed as the maxima between the previously
computed temporal relaxations for the left and the right
subformulae.

Note that considering primitive formulae in Alg.~\ref{alg:policy},
instead of traversing the AST all the way to the leaves, optimizes
the running time and the level of recursion of the algorithm.

A very important property of Alg.~\ref{alg:policy}
is that its complexity does not depend on the deadlines associated
with the {\em within} operators of the TWTL specification
formula $\phi$. This is an immediate consequence of the
DFA construction proposed in Sec~\ref{sec:construction}.
Moreover, it follows from Remark~\ref{rem:optimal2satisfaction}
that the completeness with respect to $\phi$ (unrelaxed)
may also be decided independently of the values of the
deadline values.
Formally, we have the following results.

\begin{theorem}
\label{th:tr-optimality}
Let $\phi$ be a TWTL formula and $\TS$ be a finite transition system.
Synthesizing a trajectory $\BF{x}$ of $\TS$ such that $\BF{o} \models \phi(\BS{\tau})$
and $\trel{\BS{\tau}}$ is minimized can be performed in
$O(\card{\phi} \cdot \card{\Delta} \cdot \card{\delta_{\FA_\infty}})$,
where $\BS{\tau} \in \BB{Z}^m$, $m$ is the number of {\em within} operators in $\phi$,
$\FA_\infty$ is the annotated DFA corresponding to $\phi(\infty)$,
$\BF{o}$ is the output trajectory induced by $\BF{x}$,
and $\Delta$ and $\delta_{\FA_\infty}$ are the sets
of transitions of $\TS$ and $\FA_\infty$, respectively.
\end{theorem}
\begin{proof}
The worst-case complexity of %the control policy synthesis algorithm
Alg.~\ref{alg:policy} is achieved when the TWTL formula $\phi$
has the form of primitive formulae enclosed by {\em within} operators
and then composed by either the conjunction, disjunction, and
concatenation operators.

The recursive algorithm stops when it encounters the primitive
formulae and executes Dijkstra's algorithm that takes
at most $O(\card{\Delta_\PA}) = O(\card{\Delta}\cdot\card{\delta_{\FA_{\infty}}})$
time. Since the recursion is performed with respect to an AST $T$
of $\phi$, the algorithm processes each operator only once.
The complexity bound follows because the size of
the set of paths $M$ returned by the algorithm is at most
the sum of the sized of the sets corresponding to the left
and the right sets $M_L$ and $M_R$, respectively.
Thus, we obtain the bound
$O(\card{\phi} \cdot \card{\Delta} \cdot \card{\delta_{\FA_\infty}})$
by summing up the time complexity over all nodes of $T$.
\end{proof}

\begin{corollary}
\label{th:completeness}
Let $\phi$ be a TWTL formula and $\TS$ be a finite transition system.
Deciding if there exists a trajectory $\BF{x}$ of $\TS$ such that
$\BF{o} \models \phi$ can be performed in
$O(\card{\phi} \cdot \card{\Delta} \cdot \card{\delta_{\FA_\infty}})$,
where $\FA_\infty$ is the annotated DFA corresponding to $\phi$,
$\BF{o}$ is the output trajectory induced by $\BF{x}$,
and $\Delta$ and $\delta_{\FA_\infty}$ are the sets
of transitions of $\TS$ and $\FA_\infty$, respectively.
\end{corollary}
\begin{proof}
It follows from Thm.~\ref{th:tr-optimality} and
Remark~\ref{rem:optimal2satisfaction}.
\end{proof}

\subsection{Verification}
\label{sec:verification}

The procedure described in Alg.~\ref{alg:verification} solves
the verification problem of a transition system $\TS$ against
all relaxed versions of a TWTL specification
First, the annotated DFA $\FA_\infty$ corresponding to $\phi$
is computed (line 1). Then a trap state $\trap$ is added in
line 2 (see Alg.~\ref{alg:disjunction} for details).
The transition system $\TS$ is composed with the DFA
$\FA_\infty$ to produce the product automaton $\PA$ (line 3).
Lastly, it is checked if a state in $\PA$ reachable from the initial
state $p_0$ exists such that its DFA component is the trap state
$\trap$ (lines 4-5).

\begin{algorithm}
\caption{Verification}
\label{alg:verification}
\DontPrintSemicolon
\KwIn{$\TS$ -- transition system}
\KwIn{$\phi$ -- TWTL specification}
\KwOut{Boolean value}
\BlankLine

$\FA_\infty \asgn translate(\phi; \inf\ = \True)$\;
add trap state $\trap$ to $\FA_\infty$\;
$\PA \asgn \TS \times \FA_\infty$\;
\lIf{$\exists x\in X$ s.t. $ p_0 \ra_\PA (x, \trap)$}{
    \Return{$\False$}\;
}\lElse{
    \Return{$\True$}\;
}
\end{algorithm}

\subsection{Learning deadlines from data}
\label{sec:learning}
In this section, we present a simple heuristic procedure to
infer deadlines from a finite set of labeled traces
such that the misclassification rate is minimized.
Let $\phi$ be a TWTL formula and $\CA{L}_p$ and $\CA{L}_n$ be
two finite sets of words labeled as positive and negative examples,
respectively.
The misclassification rate is
$\card{\{w \in L_p\ |\ w\not\models\phi(\BS{\tau})\}} + \card{\{w \in L_n\ |\ w\models\phi(\BS{\tau})\}}$,
where $\phi(\BS{\tau})$ is a feasible $\BS{\tau}$-relaxation of $\phi$.
The terms of the misclassification rate are the false negative and
false positive rates, respectively.

The procedure presented in Alg.~\ref{alg:parameter-synthesis}
uses Alg.~\ref{alg:compute-temporal-relaxation}
to compute the tightest deadlines for each trace.
Then each deadline is determined in a greedy
way such that the misclassification rate is minimized.
The heuristic in Alg.~\ref{alg:compute-temporal-relaxation}
is due to the fact that each deadline is considered
separately from the others.
However, the deadlines are not independent with
respect to the minimization of the misclassification rate.

Notice that the algorithm constructs $\FA_\infty$ only once
at line 1. Then the automaton is used in the $tr(\cdot)$ function
to compute the temporal relaxation of each trace, lines 2-3.
Thus, the procedure avoids building $\FA_\infty$ for each trace.

\begin{algorithm}
\caption{Parameter learning}
\label{alg:parameter-synthesis}
\DontPrintSemicolon
\KwIn{$\CA{L}_p$ -- set of positive traces}
\KwIn{$\CA{L}_n$ -- set of negative traces}
\KwIn{$\phi$ -- template TWTL formula}
\KwOut{$d$ -- the vector of deadlines}
\BlankLine

$\FA_\infty \asgn translate(\phi; \inf\ = \True)$\;
$D_p \asgn \{ tr(p, \FA_\infty) + \BF{b} \ |\ p \in \CA{L}_p \}$\;
$D_n \asgn \{ tr(p, \FA_\infty) + \BF{b} \ |\ p \in \CA{L}_n \}$\;
$\BF{d} \asgn (-\infty, -\infty, \ldots, -\infty)$ \tcp*[h]{$m$-dimensional}\;
\For{$k \in \{1,\ldots,m\}$}{
    $D_k \asgn \{ d[k] \ |\ d \in D_p\}$\;
    $\BF{d}[k] \asgn \arg\min_{d \in D_k}\big(\card{D^k_{FP}(d)} + \card{D^k_{FN}(d)}\big)$, where\;
    \quad $D^k_{FP}(d) \asgn \{ \BF{d}'[k] \ |\ \BF{d}'[k] > d, \BF{d}' \in D_n\}$\;
    \quad $D^k_{FN}(d) \asgn \{ \BF{d}'[k] \ |\ \BF{d}'[k] \leq d, \BF{d}' \in D_p\}$\;
}
\Return{$d$}
\end{algorithm}

In Alg.~\ref{alg:parameter-synthesis}, $m$ denotes
the number of {\em within} operators and $\BF{b}$ is
the $m$-dimensional vector of deadlines associated
with each {\em within} operator in the TWTL formula
$\phi$. We assume that the order of the {\em within}
operators is given by the post-order traversal of $AST(\phi)$,
i.e., recursively traversing the children nodes first and
then the node itself.

The complexity of the learning procedure is
$O\big(\spow{\card{\phi}+\card{\AP}} + 
(\card{\CA{L}_p} + \card{\CA{L}_n})\cdot l_m \cdot \card{\phi} +
m\cdot (\card{\CA{L}_p} + \card{\CA{L}_n})\big)$, where:
(a) the first term is the complexity of constructing $\FA_\infty$ (line 1);
(b) the second term corresponds to computing the tight deadlines
for all traces positive and negative in lines 2 and 3, respectively;
(c) the third term is the complexity of the for loop, which
computes each deadline separately in a greedy fashion (lines 5-9).
The maximum length of a trace (positive or negative) is denoted
by $l_m$ in the complexity formula.

\section{TWTL Python Package}
\label{sec:py-twtl}

We provide a Python 2.7 implementation named PyTWTL of
the proposed algorithms based on LOMAP~\cite{ulusoy-ijrr2013}, 
ANTLRv3~\cite{Parr07} and networkx~\cite{nx} libraries.
PyTWTL implementation is released under the GPLv3 license
and can be downloaded from \href{http://hyness.bu.edu/twtl}{\it hyness.bu.edu/twtl}.
The library may be used to:
\begin{enumerate}
    \item construct a DFA $\FA_\phi$ and a annotated DFA $\FA_\infty$
from a TWTL formula $\phi$;
    \item monitor the satisfaction of a TWTL formula $\phi$;
    \item monitor the satisfaction of an arbitrary relaxation of $\phi$,
    i.e., $\phi(\infty)$;
    \item compute the temporal relaxation of a trace with respect to
    a TWTL formula;
    \item compute a satisfying control policy with respect to a TWTL
    formula $\phi$;
    \item compute a minimally relaxed control policy with respect to
    a TWTL formula $\phi$, i.e., for $\phi(\BS{\tau})$ such that
    $\trel{\BS{\tau}}$ is minimal.
\end{enumerate}

The parsing of TWTL formulae is performed using ANTLRv3 framework.
We provide grammar files which may be used to port to generate
lexers and parsers for other programming languages such as
Java, C/C++, Ruby.
To support Python 2.7, we used version 3.1.3 of ANTLRv3 and the
corresponding Python runtime ANTLR library, which we included in
our distribution for convenience.

\section{Case Studies}
\label{sec:case-studies}

In this section, we present some examples highlighting the solutions
for the verification, synthesis and learning problems.
First, we show the automaton construction procedure on a
TWTL formula and how the tight deadlines are inferred for
a given trace.
Then, we consider an example involving a robot whose
motion is modeled as a TS.
The policy computation algorithm is used to solve
a path planning problem with rich specifications given as
TWTL formulae.
The procedure for performing verification,
i.e., all robot trajectories satisfy a given TWTL specification,
is also shown.
Finally, the performance of the heuristic learning algorithm is
demonstrated on a simple example.

\subsection{Automata Construction and Temporal Relaxation}

Consider the following TWTL specification over the set
of atomic propositions $\AP=\{A, B, C, D\}$:
\begin{equation}
\label{eq:case-spec}
\phi = [H^2 A]^{[0, 6]} \cdot ([H^1 B]^{[0, 3]} \orltl [H^1 C]^{[1, 4]}) \cdot [H^1 D]^{[0, 6]}
\end{equation}

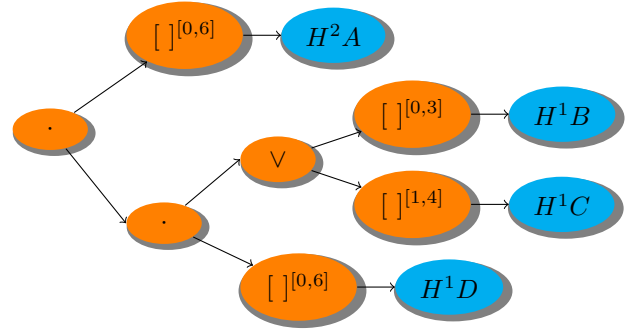
\begin{figure}[!htb]
\centering
\begin{tikzpicture}[
    grow = right,
    state/.style={ellipse, draw=none, fill=orange, circular drop shadow,
        text centered, anchor=west, text=black, minimum width=1cm},
    leaf/.style={ellipse, draw=none, fill=cyan, circular drop shadow,
        text centered, anchor=west, text=black},
    level distance=0.5cm, growth parent anchor=east
]

\node [state] {$\cdot$} [->, sibling distance=2.5cm]
     child { [->, sibling distance=1.7cm]
        node [state] {$\cdot$}
        child {
            node [state] {$[\ ]^{[0, 6]}$}
            child {
                node [leaf] {$H^1 D$}
            }
        }
        child { [->, sibling distance=1.2cm]
            node [state] {$\orltl$}
            child {
                node [state] {$[\ ]^{[1, 4]}$}
                child {
                    node [leaf] {$H^1 C$}
                }
            }
            child {
                node [state] {$[\ ]^{[0, 3]}$}
                child {
                    node [leaf] {$H^1 B$}
                }
            }
        }
     }
     child {
        node [state] {$[\ ]^{[0,6]}$}
        child{
            node [leaf] {$H^2 A$}
        }
     }
;
\end{tikzpicture}
\caption{The AST corresponding to the TWTL formula in Eq.~\eqref{eq:case-spec}.}
\label{fig:ast-case-study}
\end{figure}

An AST of formula $\phi$ is shown in Fig.~\ref{fig:ast-case-study}.
The TWTL formula $\phi$ is converted to an annotated DFA
$\FA_\infty$ using Alg.~\ref{alg:twtl2dfa}.
The procedure recursively constructs the DFA from the leafs of
the AST to the root.
A few processing steps are shown in
Fig.~\ref{fig:dfa-construction-case-study}.
The construction of DFA corresponding to a leaf, i.e., a {\em hold}
operator, is straightforward, see Fig.~\ref{fig:case-study-leaf}.
Next, the transformation corresponding to a {\em within}
operator is shown in Fig.~\ref{fig:case-study-within}.
Note that the delay of one time unit is due to the lower bound of the
time window of the {\em within} operator.
Also, note that the automaton restarts on symbols that
block the DFA corresponding to the inner formula $H^1 C$.

\begin{figure*}[p]
\centering
\begin{subfigure}[c]{0.48\textwidth}
\centering
\begin{tikzpicture}[->,>=stealth',shorten >=1pt,auto,node distance=1.8cm, semithick,
initial text={}]
\tikzstyle{every state}=[text=black]

\node[initial,state] (s0)                    {$s_0$};
\node[state]          (s1) [right of=s0] {$s_1$};
\node[state]          (s2) [right of=s1] {$s_2$};
\node[state,accepting] (s3) [right of=s2] {$s_3$};

\path[->] (s0) edge              node {$A$} (s1)
              (s1) edge              node {$A$} (s2)
              (s2) edge              node {$A$} (s3);
\end{tikzpicture}
\caption{$H^2 A$}
\label{fig:case-study-leaf}
\end{subfigure}
~%
\begin{subfigure}[c]{0.48\textwidth}
\centering
\begin{tikzpicture}[->,>=stealth',shorten >=1pt,auto,node distance=1.8cm, semithick,
initial text={}]
\tikzstyle{every state}=[text=black]

\node[initial,state] (s0)                    {$s_0$};
\node[state]          (s1) [right of=s0] {$s_1$};
\node[state]          (s2) [right of=s1] {$s_2$};
\node[state,accepting] (s3) [right of=s2] {$s_3$};

\path[->] (s0) edge                    node {$\True$} (s1)
              (s1) edge                    node {$C$} (s2)
                     edge [loop above] node {$\notltl C$} (s1)
              (s2) edge                    node {$C$} (s3)
                     edge [bend left] node {$\notltl C$} (s1);
\end{tikzpicture}
\caption{$[H^1 C]^{[1,\ ]}$}
\label{fig:case-study-within}
\end{subfigure}
\\ \bigskip%~%
\begin{subfigure}[c]{0.9\textwidth}
\centering
\begin{tikzpicture}[->,>=stealth',shorten >=1pt,auto,node distance=2.5cm, semithick,
initial text={}]
\tikzstyle{every state}=[text=black]

\node[initial,state]      (s00)                                {$s_{00}$};
\node[state]               (s11) [above right of=s00] {$s_{11}$};
\node[state,accepting] (s22) [right of=s11,node distance=3.5cm] {$s_{22}$};
\node[state]               (s01) [below right of=s00] {$s_{01}$};
\node[state]               (s02) [right of=s01,node distance=3cm]          {$s_{02}$};
\node[state,accepting] (s13) [right of=s02, node distance=4.5cm] {$s_{13}$};
\node[state]               (s12) [right of=s11, node distance=5.5cm] {$s_{12}$};
\node[state,accepting] (s23) [right of=s12]          {$s_{23}$};
\node[state,accepting] (s21) [above right of=s12] {$s_{21}$};
\node[state,accepting] (s03) [below right of=s12] {$s_{03}$};

\path[->] (s00) edge [font=\scriptsize,sloped,above] node {$B$}                               (s11)
                       edge [font=\scriptsize,sloped,below] node {$\notltl B$}                      (s01)
              (s11) edge [font=\scriptsize]                     node {$B\andltl \notltl C$}          (s21)
                       edge [font=\scriptsize,sloped,below] node {$B\andltl C$}                    (s22)
                      edge [font=\scriptsize,bend left=10,sloped,above] node {$\notltl B\andltl \notltl C$} (s01)
                      edge [font=\scriptsize,bend left=10,sloped,above] node {$\notltl B\andltl C$} (s02)
              (s01) edge [font=\scriptsize,loop left] node                   {$\notltl B\andltl \notltl C$}          (s01)
                      edge [font=\scriptsize,bend left=10,sloped,above] node {$B\andltl \notltl C$}          (s11)
                      edge [font=\scriptsize,bend left=10,sloped,above] node {$B\andltl C$}          (s12)
                      edge [font=\scriptsize] node                   {$\notltl B\andltl C$}          (s02)
             (s02) edge [font=\scriptsize,bend left] node                   {$\notltl B\andltl \notltl C$}          (s01)
                     edge [font=\scriptsize,bend left=10,sloped,above] node {$B\andltl \notltl C$}          (s11)
                     edge [font=\scriptsize] node {$\notltl B \andltl C$} (s03)
                     edge [font=\scriptsize,sloped,above] node {$B \andltl C$} (s13)
             (s12) edge [font=\scriptsize,bend left=10,sloped,above] node                   {$\notltl B\andltl \notltl C$}          (s01)
                     edge [font=\scriptsize] node {$B \andltl C$} (s23)
                     edge [font=\scriptsize,sloped,above] node {$B \andltl \notltl C$} (s21)
                     edge [font=\scriptsize,sloped,above] node {$\notltl B \andltl C$} (s03);
%              (s2) edge                    node {$C$} (s3)
%                     edge [bend left] node {$\notltl C$} (s1);
\end{tikzpicture}
\caption{$[H^1 B]^{[0,\ ]} \orltl [H^1 C]^{[1,\ ]}$}
\label{fig:case-study-disjunction-exp}
\end{subfigure}
\\ \bigskip%~%
\begin{subfigure}[c]{0.9\textwidth}
\centering
\begin{tikzpicture}[->,>=stealth',shorten >=1pt,auto,node distance=2.5cm, semithick,
initial text={}]
\tikzstyle{every state}=[text=black]

\node[initial,state]      (s00)                                {$s_{00}$};
\node[state]               (s11) [above right of=s00] {$s_{11}$};
%\node[state,accepting] (sff) [right of=s11,node distance=3.5cm] {$s_{f}$};
\node[state]               (s01) [below right of=s00] {$s_{01}$};
\node[state]               (s02) [right of=s00,node distance=4cm]          {$s_{02}$};
\node[state]               (s12) [right of=s01, node distance=4cm] {$s_{12}$};
\node[state,accepting] (sff) [right of=s11, node distance=4cm] {$s_{f}$};

\path[->] (s00) edge [font=\scriptsize,sloped,above] node {$B$}                                                 (s11)
                       edge [font=\scriptsize,sloped,below] node {$\notltl B$}                                        (s01)
              (s11) edge [font=\scriptsize]                     node {$B$}                                                  (sff)
                      edge [font=\scriptsize,bend left=10,sloped,above] node {$\notltl B\andltl \notltl C$} (s01)
                      edge [font=\scriptsize,bend left=10,sloped,above] node {$\notltl B\andltl C$}          (s02)
              (s01) edge [font=\scriptsize,loop left] node                   {$\notltl B\andltl \notltl C$}         (s01)
                      edge [font=\scriptsize,bend left=10,sloped,above] node {$B\andltl \notltl C$}          (s11)
                      edge [font=\scriptsize,bend left=10,sloped,above] node {$B\andltl C$}                   (s12)
                      edge [font=\scriptsize,bend left=10,sloped,above] node {$\notltl B\andltl C$}          (s02)
             (s02) edge [font=\scriptsize,bend left=10,sloped,below] node {$\notltl B\andltl \notltl C$} (s01)
                     edge [font=\scriptsize,bend left=10,sloped,below] node {$B\andltl \notltl C$}            (s11)
                     edge [font=\scriptsize,sloped,above] node {$C$}                                                    (sff)
             (s12) edge [font=\scriptsize,bend left=10,sloped,above] node {$\notltl B\andltl \notltl C$}  (s01)
                     edge [font=\scriptsize,sloped,above] node {$B \orltl C$}                                         (sff);
\end{tikzpicture}
\caption{$[H^1 B]^{[0,\ ]} \orltl [H^1 C]^{[1,\ ]}$}
\label{fig:case-study-disjunction}
\end{subfigure}
\\ \bigskip%~%
\begin{subfigure}[c]{0.9\textwidth}
\centering
\begin{tikzpicture}[->,>=stealth',shorten >=1pt,auto,node distance=2.5cm, semithick,
initial text={}]
\tikzstyle{every state}=[text=black]

\node[initial,state]       (s0)                                                      {$s_0$};
\node[state]               (s1) [right of=s0,node distance=1.5cm]    {$s_1$};
\node[state]               (s2) [right of=s1,node distance=1.5cm]    {$s_2$};
\node[state]               (s3) [right of=s2,node distance=1.5cm]    {$s_3$};
\node[state]               (s4) [above right of=s3]                          {$s_4$};
\node[state]               (s5) [below right of=s3]                          {$s_5$};
\node[state]               (s6) [right of=s3,node distance=4cm]      {$s_6$};
\node[state]               (s7) [right of=s5, node distance=4cm]     {$s_7$};
\node[state]               (s8) [right of=s4, node distance=4cm]     {$s_8$};
\node[state]               (s9) [below right of=s8]                         {$s_9$};
\node[state,accepting] (s10) [right of=s9,node distance=1.5cm] {$s_{10}$};

\path[->] (s0) edge                       node {$A$}          (s1)
                     edge [loop above]    node {$\notltl A$} (s0)
              (s1) edge                       node {$A$}          (s2)
                     edge [bend left=20] node {$\notltl A$} (s0)	
              (s2) edge                       node {$A$}          (s3)
                     edge [bend left=50] node {$\notltl A$} (s0)
              (s3) edge [font=\scriptsize,sloped,above] node {$B$}                                                  (s4)
                     edge [font=\scriptsize,sloped,below] node {$\notltl B$}                                         (s5)
              (s4) edge [font=\scriptsize]                     node {$B$}                                                   (s8)
                     edge [font=\scriptsize,bend left=10,sloped,above] node {$\notltl B\andltl \notltl C$} (s5)
                     edge [font=\scriptsize,bend left=10,sloped,above] node {$\notltl B\andltl C$}          (s6)
              (s5) edge [font=\scriptsize,loop left] node                   {$\notltl B\andltl \notltl C$}         (s5)
                     edge [font=\scriptsize,bend left=10,sloped,above] node {$B\andltl \notltl C$}          (s4)
                     edge [font=\scriptsize,bend left=10,sloped,above] node {$B\andltl C$}                   (s7)
                     edge [font=\scriptsize,bend left=10,sloped,above] node {$\notltl B\andltl C$}          (s6)
             (s6) edge [font=\scriptsize,bend left=10,sloped,below] node {$\notltl B\andltl \notltl C$}  (s5)
                    edge [font=\scriptsize,bend left=10,sloped,below] node {$B\andltl \notltl C$}           (s4)
                    edge [font=\scriptsize,sloped,above] node {$C$}                                                   (s8)
             (s7) edge [font=\scriptsize,bend left=10,sloped,above] node {$\notltl B\andltl \notltl C$}  (s5)
                    edge [font=\scriptsize,sloped,above] node {$B \orltl C$}                                        (s8)
             (s8) edge [font=\scriptsize,bend left=10,sloped,above] node {$D$}                                (s9)
                    edge [loop right]    node {$\notltl D$}                                                                   (s8)
             (s9) edge [font=\scriptsize,sloped,above] node {$D$}                                                   (s10)
                    edge [font=\scriptsize,bend left=10,sloped,below] node {$\notltl D$}                      (s8);
\end{tikzpicture}
\caption{$[H^2 A]^{[0,\ ]} \cdot ([H^1 B]^{[0,\ ]} \orltl [H^1 C]^{[1,\ ]}) \cdot [H^1 D]^{[0,\ ]}$}
\label{fig:case-study-spec-dfa}
\end{subfigure}
%
%
%\begin{tikzpicture}[
%    grow = right,
%    state/.style={ellipse, draw=none, fill=orange, circular drop shadow,
%        text centered, anchor=west, text=black, minimum width=1cm},
%    leaf/.style={ellipse, draw=none, fill=cyan, circular drop shadow,
%        text centered, anchor=west, text=black},
%    level distance=0.5cm, growth parent anchor=east
%]
%
%\node [state] {$\cdot$} [->, sibling distance=2.5cm]
%     child { [->, sibling distance=1.7cm]
%        node [state] {$\cdot$}
%        child {
%            node [state] {$[\ ]^{[0, 6]}$}
%            child {
%                node [leaf] {$H^1 D$}
%            }
%        }
%        child { [->, sibling distance=1.2cm]
%            node [state] {$\orltl$}
%            child {
%                node [state] {$[\ ]^{[1, 4]}$}
%                child {
%                    node [leaf] {$H^1 C$}
%                }
%            }
%            child {
%                node [state] {$[\ ]^{[0, 3]}$}
%                child {
%                    node [leaf] {$H^1 B$}
%                }
%            }
%        }
%     }
%     child {
%        node [state] {$[\ ]^{[0,6]}$}
%        child{
%            node [leaf] {$H^2 A$}
%        }
%     }
%;
%\end{tikzpicture}
\caption{Annotated automata corresponding to subformulae of the TWTL specification in Eq.~\eqref{eq:case-spec}.}
\label{fig:dfa-construction-case-study}
\end{figure*}
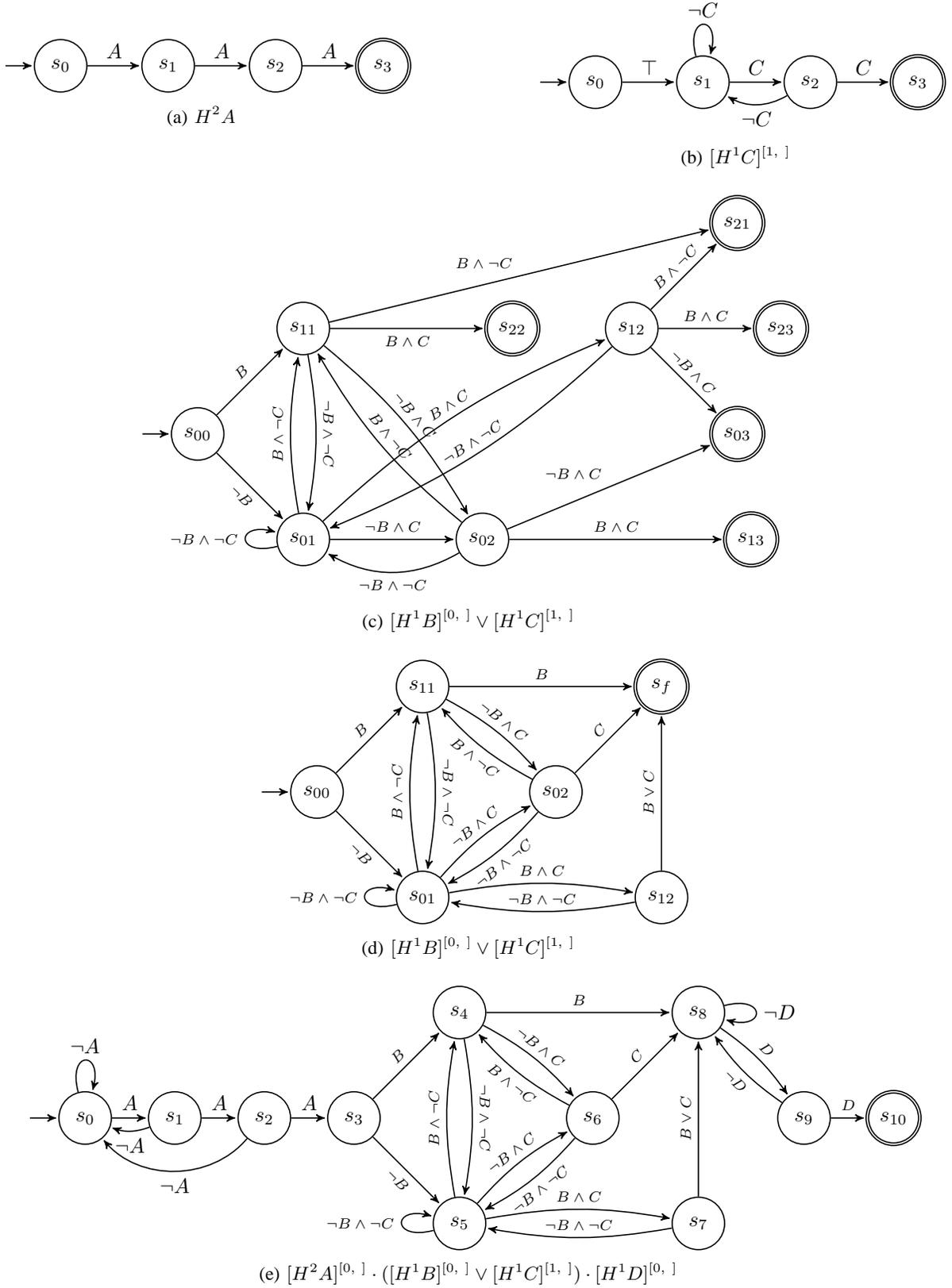

The next two figures, Fig.~\ref{fig:case-study-disjunction-exp}
and Fig.~\ref{fig:case-study-disjunction}, show the translation of
the disjunction operator.
Specifically, Fig.~\ref{fig:case-study-disjunction-exp}, shows the
product DFA corresponding to the disjunction without merging the final states.
Since none of the final states have outgoing transitions,
see Corr.~\ref{th:unambiguous-dfa-final}, and they can
be merged into a single final state, see
Fig.~\ref{fig:case-study-disjunction}.
However, we still need to keep track of which subformula of
the disjunctions holds.
The annotation variable $T.choices$, introduced in
Sec.~\ref{sec:op-conj-disj}, stores this information as
{\small
\begin{equation}
\label{eq:case-study-choices}
%T.choices =
\begin{cases}
L=\{(s_{11}, B\andltl \notltl C), (s_{11}, B\andltl C), (s_{12}, B \andltl \notltl C)\},\\
R=\{(s_{02}, \notltl B \andltl C), (s_{02}, B \andltl C), (s_{12}, \notltl B \andltl C)\},\\
B=\{(s_{12}, B\andltl C)\}.
\end{cases}
\end{equation}
}%
Notice that the tuples in
Eq.~\eqref{eq:case-study-choices} correspond to the
ingoing edges of the final states in the DFA from
Fig.~\ref{fig:case-study-disjunction-exp}. Finally, the DFA corresponding to the overall specification
formula $\phi$ is shown in Fig.~\ref{fig:case-study-spec-dfa}.

Let $\phi_A = [H^2 A]^{[0, 6]}$, $ \phi_B = [H^1 B]^{[0, 3]}$,
$\phi_C = [H^1 C]^{[1, 4]}$, and $\phi_D = [H^1 D]^{[0, 6]}$ be
subformulae of $\phi$ associated with the {\em within} operators.
The annotation data for these subformulae is shown in the following table.

\begin{center}
%\resizebox{0.98\linewidth}{!}{
\begin{tabular}{ c c c }
Subformula & $T.I$  & $T.F$ \\ \hline\hline
$\phi$ & $\{s_0\}$ & $\{s_{10}\}$ \\
$\phi_A$ & $\{s_0\}$ & $\{s_3\}$ \\
$\phi_B$ & $\{s_3, s_5, s_6\}$ & $\{s_8\}$ \\
$\phi_C$ & $\{s_3\}$ & $\{s_3\}$ \\
$\phi_D$ & $\{s_8\}$ & $\{s_{10}\}$ \\
\end{tabular}
%}
\end{center}

Consider the following word over the alphabet
$\Sigma=\spow{\AP}$:
\begin{equation}
\label{eq:case-study-word-both}
\BS{\sigma} = \epsilon, \{A\}, \{A\}, \{A\}, \epsilon, \{B, C\}, \{B, C\}, \epsilon, \{D\}, \{D\}\,
\end{equation}
where $\epsilon$ is the empty symbol.
%and the singleton sets
%were written without curly brackets to simplify the notation.
The following table shows the stages of
Alg.~\ref{alg:compute-temporal-relaxation} as the symbols
of the word $\BS{\sigma}$ are processed:
\begin{center}
\resizebox{0.98\linewidth}{!}{
\begin{tabular}{ c c c c c c c }
No. & Symbol         & State  & $\phi_A$                & $\phi_B$                  & $\phi_C$                    & $\phi_D$ \\ \hline\hline
Init&                       & $s_0$ & $(\True, \False, 0)$ & $(\False, \False, -1)$ & $(\False, \False, -1)$  & $(\False, \False, -1)$\\
0 & $\epsilon$       & $s_0$ & $(\True, \False, 1)$ & $(\False, \False, -1)$ & $(\False, \False, -1)$  & $(\False, \False, -1)$\\
1 & $\{A\}$           & $s_1$ & $(\True, \False, 2)$ & $(\False, \False, -1)$ & $(\False, \False, -1)$  & $(\False, \False, -1)$\\
2 & $\{A\}$           & $s_2$ & $(\True, \False, 3)$ & $(\False, \False, -1)$ & $(\False, \False, -1)$  & $(\False, \False, -1)$\\
3 & $\{A\}$           & $s_3$ & $(\False, \True, 3)$ & $(\True, \False,   0)$ & $(\True, \False,   0)$  & $(\False, \False, -1)$\\
4 & $\epsilon$       & $s_5$ & $(\False, \True, 3)$ & $(\True, \False,   1)$ & $(\True, \False,   1)$  & $(\False, \False, -1)$\\
5 & $\{B, C\}$       & $s_7$ & $(\False, \True, 3)$ & $(\True, \False,   2)$ & $(\True, \False,   2)$  & $(\False, \False, -1)$\\
6 & $\{B, C\}$       & $s_8$ & $(\False, \True, 3)$ & $(\False, \True,   2)$ & $(\False, \True,   2)$  & $(\True, \False,  0)$\\
7 & $\epsilon$       & $s_8$ & $(\False, \True, 3)$ & $(\False, \True,   2)$ & $(\False, \True,   2)$  & $(\True, \False,  1)$\\
8 & $\{D\}$           & $s_9$ & $(\False, \True, 3)$ & $(\False, \True,   2)$ & $(\False, \True,   2)$  & $(\True, \False,  2)$\\
9 & $\{D\}$        &$s_{10}$& $(\False, \True, 3)$ & $(\False, \True,   2)$ & $(\False, \True,   2)$  & $(\False, \True,  2)$\\
\end{tabular}
}
\end{center}
where each 3-tuple in last four columns represents the
annotation variables $T.ongoing$, $T.done$ and $T.steps$,
respectively.
The temporal relaxation for $\BS{\sigma}$ can be extracted
from the values in the last row by subtracting the deadlines
of the {\em within} operators from them.
Thus, the vector of tightest $\tau$ values is $(-3, -1, -2, -3)$.
However, because $\phi_B$ and $\phi_C$ are in disjunction,
we have the temporal relaxation $\BS{\tau} = (-3, -\infty, -2, -3)$,
where we choose to ignore the subformula containing $\phi_B$.
Thus, the maximum temporal relaxation is $\trel{\BS{\tau}} = -2$.

\subsection{Control Policy Synthesis}
\label{sec:control-case-study}

Consider a robot moving in an environment represented as the finite
graph shown in Fig.~\ref{fig:control-case-study-env}.
The nodes of the graph represent the points of interest, while the edges
indicate the possibility of moving the robot between the edges' endpoints.
The numbers associated with the edges represent the travel times, and we
assume that all the travel times are integer multiples of a time step
$\Delta t$.
The robot may also stay at any of the points of interest.

The motion of the robot is abstracted as a transition system
$\TS$, which is obtained from the finite graph by splitting each edge into a
number of transitions equal to the corresponding edge's travel time.
The generated transition system thus has 27 states and 67 transitions and
is shown in Fig.~\ref{fig:control-case-study-ts}.

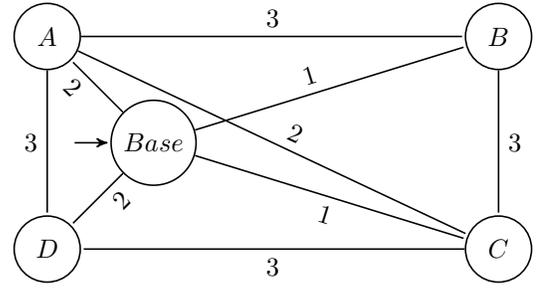
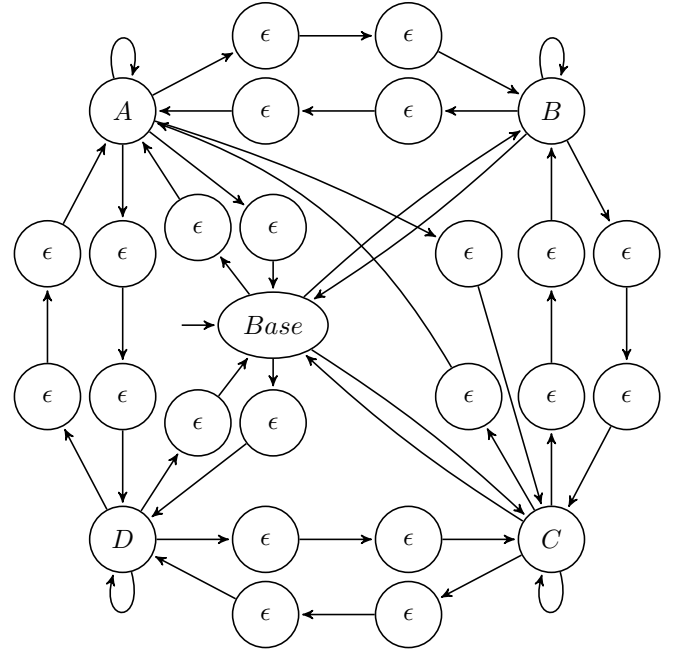
\begin{figure}[!htb]
\centering
\begin{subfigure}[c]{\linewidth}
\centering
\begin{tikzpicture}[->,>=stealth',shorten >=1pt,auto,node distance=2cm, semithick,
initial text={}]
\tikzstyle{every state}=[text=black]

\node[initial,state] (Base)                       {$Base$};
\node[state]         (A) [above left of=Base] {$A$};
\node[state]         (D) [below left of=Base] {$D$};
\node[state]         (B) [right of=A,node distance=6cm] {$B$};
\node[state]         (C) [right of=D,node distance=6cm] {$C$};

\path[-]
(Base) edge [sloped] node {$2$} (A)
(Base) edge [sloped] node {$1$} (B)
(Base) edge [below,sloped]  node {$1$} (C)
(Base) edge [sloped] node {$2$} (D)
(A) edge              node {$3$} (B)
(A) edge [sloped] node {$2$} (C)
(A) edge [left]      node {$3$} (D)
(B) edge              node {$3$} (C)
(C) edge              node {$3$} (D)
;
\end{tikzpicture}
\caption{An environment with five points of interest, $Base$ $A$, $B$, $C$, and $D$.
The edges indicate the existence of paths between their endpoints, while the associated
numbers represent the travel times of the edges. The robot may stay at a region of
interest.}
\label{fig:control-case-study-env}
\end{subfigure}
\\\medskip
\begin{subfigure}[c]{\linewidth}
\centering
\begin{tikzpicture}[->,>=stealth',shorten >=1pt,auto,node distance=1.9cm, semithick,
initial text={}]
\tikzstyle{every state}=[text=black]

\node[initial,state] (Base) [xshift=2cm,yshift=-2.85cm,ellipse]                       {$Base$};
%\node[state]         (BaseA) [above left of=Base,xshift=-0.5cm,yshift=-0.5cm] {$\epsilon$};
\node[state]         (A) {$A$};
%\node[state]         (ABase) [above left of=Base,xshift=0.5cm,yshift=0.5cm] {$\epsilon$};
\foreach \name/\text/\xoff/\yoff/\prev in {AB1/\epsilon/0cm/1cm/A, AB2/\epsilon/0cm/0cm/AB1}
    \node[state,xshift=\xoff,yshift=\yoff,right of=\prev] (\name) {$\text$};
\foreach \name/\text/\xoff/\yoff/\prev in {AD1/\epsilon/0cm/0cm/A, AD2/\epsilon/0cm/0cm/AD1}
    \node[state,xshift=\xoff,yshift=\yoff,below of=\prev] (\name) {$\text$};
\node[state]         (B) [yshift=-1cm,right of=AB2] {$B$};
\foreach \name/\text/\xoff/\yoff/\prev in {BA1/\epsilon/0cm/0cm/B, BA2/\epsilon/0cm/0cm/BA1}
    \node[state,xshift=\xoff,yshift=\yoff,left of=\prev] (\name) {$\text$};
\foreach \name/\text/\xoff/\yoff/\prev in {BC1/\epsilon/1cm/0cm/B, BC2/\epsilon/0cm/0cm/BC1}
    \node[state,xshift=\xoff,yshift=\yoff,below of=\prev] (\name) {$\text$};
\node[state]         (C) [below of=BC2,xshift=-1cm] {$C$};
\foreach \name/\text/\xoff/\yoff/\prev in {CB1/\epsilon/0cm/0cm/C, CB2/\epsilon/0cm/0cm/CB1}
    \node[state,xshift=\xoff,yshift=\yoff,above of=\prev] (\name) {$\text$};
\foreach \name/\text/\xoff/\yoff/\prev in {CD1/\epsilon/0cm/-1cm/C, CD2/\epsilon/0cm/0cm/CD1}
    \node[state,xshift=\xoff,yshift=\yoff,left of=\prev] (\name) {$\text$};
\node[state]         (D) [left of=CD2,yshift=1cm] {$D$};
\foreach \name/\text/\xoff/\yoff/\prev in {DC1/\epsilon/0cm/0cm/D, DC2/\epsilon/0cm/0cm/DC1}
    \node[state,xshift=\xoff,yshift=\yoff,right of=\prev] (\name) {$\text$};
\foreach \name/\text/\xoff/\yoff/\prev in {DA1/\epsilon/-1cm/0cm/D, DA2/\epsilon/0cm/0cm/DA1}
    \node[state,xshift=\xoff,yshift=\yoff,above of=\prev] (\name) {$\text$};

\node[state] (AC) [right of=Base,xshift=.7cm,yshift=0.95cm] {$\epsilon$};
\node[state] (CA) [below of=AC] {$\epsilon$};
\node[state] (ABase) [above of=Base,node distance=1.3cm] {$\epsilon$};
\node[state] (BaseA) [left of=ABase,node distance=1cm] {$\epsilon$};
\node[state] (BaseD) [below of=Base,node distance=1.3cm] {$\epsilon$};
\node[state] (DBase) [left of=BaseD,node distance=1cm] {$\epsilon$};

%\node[state]         (AD2) [below of=AD1] {$\epsilon$};

\path[->]
(Base) edge [bend left=5] (B)
          edge [bend left=5] (C)
          edge (BaseA)
          edge (BaseD)
(BaseA) edge (A)
(BaseD) edge (D)
(A) edge [loop above] (A)
      edge (AB1)
      edge (AD1)
      edge [bend left=5] (AC)
      edge (ABase)
(AB1) edge (AB2)
(AB2) edge (B)
(AD1) edge (AD2)
(AD2) edge (D)
(AC) edge (C)
(ABase) edge (Base)
(B) edge [loop above] (B)
     edge (BA1)
     edge (BC1)
     edge [bend left=5] (Base)
(BA1) edge (BA2)
(BA2) edge (A)
(BC1) edge (BC2)
(BC2) edge (C)
(C) edge [loop below] (C)
     edge (CB1)
     edge (CD1)
     edge (CA)
     edge [bend left=5] (Base)
(CB1) edge (CB2)
(CB2) edge (B)
(CD1) edge (CD2)
(CD2) edge (D)
(CA) edge [bend right=20] (A)
(D) edge [loop below] (D)
     edge (DA1)
     edge (DC1)
     edge (DBase)
(DA1) edge (DA2)
(DA2) edge (A)
(DC1) edge (DC2)
(DC2) edge (C)
(DBase) edge (Base)
;
\end{tikzpicture}
\caption{The transition system $\TS$ obtained from the environment graph shown in Fig.~\ref{fig:control-case-study-env}.}
\label{fig:control-case-study-ts}
\end{subfigure}
\caption{The environment where the robot operates and its abstraction $\TS$.}
\label{fig:control-case-study}
\end{figure}

Consider the TWTL specification $\phi$ from Eq.~\eqref{eq:case-spec}.
The product automaton $\PA = \TS \times \FA_\infty$ is constructed,
where $\FA_\infty$ is the annotated DFA corresponding to $\phi(\infty)$
shown in Fig.~\ref{fig:case-study-spec-dfa}.
The product automaton $\PA$ has 204 states and 378 transitions.
The control policy computed by using Alg.~\ref{alg:policy} is
\begin{equation}
\label{eq:control-case-study-policy}
\BF{x} = Base, A, A, A, C, C, Base, D, D,
\end{equation}
which generates the output word
\begin{equation}
\label{eq:control-case-study-output-word}
\BS{\sigma} = \epsilon, \epsilon, \{A\}, \{A\}, \{A\}, \epsilon, \{C\}, \{C\}, \epsilon, \epsilon, \{D\}, \{D\}.
\end{equation}
The minimum temporal relaxation for $\BS{\sigma}$ is
$\trel{\BS{\tau}} = -2$, where $\BS{\tau}= (-2, -\infty, -2, -3)$ is the
minimal temporal relaxation vector associated with $\BS{\sigma}$.

%case studies:
%\begin{itemize}
%  \item synthesize control policy for a robot $\TS$ w.r.t. (a) $\phi$; (b) $\phi(\infty)$; and (c) compute robustness.
%\end{itemize}

\subsection{Verification}
\label{sec:verification-case-study}

In the verification problem, we are concerned
with checking for the existence of relaxed specifications for
every possible run of a transition system.

To illustrate this problem, consider the transition system in
Fig.~\ref{fig:verification-case-study-ts} and the following two
TWTL specifications:
\begin{align}
\phi_1 &= [H^1 A]^{[1, 2]} \label{eq:verification-case-study-spec-true}\\
\phi_2 &= [H^1 \notltl B]^{[1, 2]} \label{eq:verification-case-study-spec-false}
\end{align}

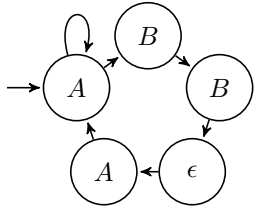
\begin{figure}[!htb]
\centering
\begin{tikzpicture}[->,>=stealth',shorten >=1pt,auto,node distance=2cm, semithick,
initial text={}]
\tikzstyle{every state}=[text=black]
\node[initial,state] (s0)  at (162:1cm) {$A$};
\foreach \name/\angle/\text in {s1/90/B, s2/18/B, s3/-54/\epsilon, s4/234/A}
    \node[state] (\name) at (\angle:1cm) {$\text$};
\foreach \from/\to in {s0/s1,s1/s2,s2/s3,s3/s4,s4/s0}
    \draw (\from) -> (\to);
\path[->] (s0) edge [loop above] (s0);
\end{tikzpicture}
\caption{A simple transition system $\TS^{simple}$.}
\label{fig:verification-case-study-ts}
\end{figure}

To check the transition system $\TS^{simple}$ against the two specifcations,
we can use Alg.~\ref{alg:verification}.
It is straightforward that the procedure will return true for $\phi_1$, because
every run of $\TS^{simple}$ satisfies $\phi_1(3)  = [H^1]^{[1, 2+3]}$.
Note that the runs of the transition system may not need to satisfy the original
specification as the satisfaction of a relaxed version is sufficient.
Similarly, Alg.~\ref{alg:verification} returns false for $\phi_2$, because
there exists a run of $\TS^{simple}$ that does not satisfy the specification,
e.g., $\BF{x} = A, B, B$.

An important conclusion highlighted by the two examples is that
the verification problem proposed in this paper is concerned with
checking a system against the logical structure of a specification and
not against any particular time bounds.
This might be useful in situation where the deadlines of
the specification are not known {\em a priori}, but the logical structure
of the specification is.

\subsection{Learning deadlines from data}
%\textcolor{blue} {is there a reason you put dot after learning?} Parameter Synthesis}
\label{sec:learning-case-study}

In the previous two cases, we use the TWTL specifications in conjunction
with problems involving infinite sets of words encoded as transition systems.
However, it is often the case that only finite sets of (output) trajectories are
available. In this section, we give a simple example of the learning problem
presented in Sec.~\ref{sec:pb-formulation}.

Consider the specification
${\phi_{learn} = [H^1]^{[0, d_1]} \cdot [H^2 B]^{[0, d_2]}}$ with unknown
deadlines and the following set of labeled trajectories, where
$C_p$ and $C_n$ are the positive and negative example labels, respectively:

\begin{equation*}
%\resizebox{0.98\linewidth}{!}{
\arraycolsep=0.1pt
\begin{array}{ *{9}{l} c c }
\multicolumn{9}{l}{\text{Word}} & \text{Label}\,  & \text{Deadlines} \\ \hline\hline
\BS{\sigma}_1=&\{A\}, &\{A\}, &\{A\}, &\{B\}, &\{B\}, &\{B\}, &\{B\}, &\epsilon & C_p & (2, 3) \\
\BS{\sigma}_2=&\epsilon, &\{A\}, &\{A\}, &\epsilon, &\{B\}, &\{B\}, &\{B\}, &\epsilon & C_p & (2, 3) \\\hline
\BS{\sigma}_3=&\{B\}, &\epsilon, &\{A\}, &\{A\}, &\{B\}, &\{B\}, &\{B\}, &\{B\} & C_n & (3, 2) \\
\BS{\sigma}_4=&\epsilon, &\{A\}, &\{A\}, &\epsilon, &\epsilon, &\{B\}, &\{B\}, &\{B\} & C_n & (2, 4)
\end{array}
%}
\end{equation*}

The last column in the above table shows the tight deadlines obtained
in lines 2 and 3 of Alg.~\ref{alg:parameter-synthesis}. Next, the learning
algorithm computes the heuristic sets $D^k_{FP}$ and $D^k_{FN}$,
$k\in \{d_1, d_2\}$, of false positive and false negative trajectories,
respectively:

\begin{center}
\setlength{\tabcolsep}{3pt}
%\resizebox{0.98\linewidth}{!}{
\begin{tabular}{ c c c c c}
Deadline & Value  & $D^k_{FP}$ & $D^k_{FN}$ & $\card{D^k_{FP}} + \card{D^k_{FN}}$ \\[0.3em] \hline\hline
$d_1$ & 2 & $\{\BS{\sigma}_4\}$ & $\emptyset$ & 1 \\
$d_1$ & 3 & $\{\BS{\sigma}_3,\BS{\sigma}_4\}$ & $\emptyset$ & 2 \\ \hline
$d_2$ & 2 & $\{\BS{\sigma}_3\}$ & $\{\BS{\sigma}_1, \BS{\sigma}_2\}$ & 3 \\
$d_2$ & 3 & $\{\BS{\sigma}_3\}$ & $\emptyset$ & 1\\
$d_2$ & 4 & $\{\BS{\sigma}_3, \BS{\sigma}_4\}$ & $\emptyset$ & 2 \\
\end{tabular}
%}
\end{center}
Finally, Alg.~\ref{alg:parameter-synthesis} chooses the deadline pair
$\BF{d} = (d_1, d_2) = (2, 3)$ that has the lowest heuristic misclassification
rate, $\card{D^k_{FP}} + \card{D^k_{FN}}$ shown in the last column of the
above table, for $d_1$ and $d_2$, respectively.
An important observation is that the inferred formula
$\phi^{\BF{d}}_{learn} = [H^1 A]^{[0, 2]}\cdot [H^2 B]^{[0, 3]}$
has zero as actual misclassification rate.
The discrepancy between the values in the table and the actual value
of the final misclassification rate are due to the heuristic of synthesizing
each deadline separately. Thus, the heuristic procedure in
Alg.~\ref{alg:parameter-synthesis} ignores the temporal and
logical structure of the template TWTL formula which may lead to
suboptimal performance, i.e., misclassification rate.

\section{Conclusion} 
\label{sec:conclusion}
In this paper, we introduced the time window temporal logic (TWTL), which
%\textcolor{blue}{
is a linear-time logic encoding sets of discrete-time bounded specifications.
%}
%is a rich language to express various time bounded specifications. 
Different from other temporal logics, TWTL has an explicit concatenation operator, which enables the compact representation of serial tasks mostly prevalent in robotics and control applications. Such a compact representation significantly reduces the complexity of constructing the automaton for the accepting language. In this paper, we also discussed the temporal relaxation of TWTL formulae and provided provably-correct algorithms to construct a compact automaton representing all temporally relaxed formulae of a given TWTL formula. Stemming from the definition of temporal relaxation, we formulated some problems related to verification, synthesis, and learning. We demonstrated the potential of TWTL and its relaxation on these problems. Finally, we also developed a Python package to solve the verification, synthesis, and learning problems.

\section{Appendices}

\subsection{Proof of Prop.~\ref{th:unambiguous-concat}}
\label{app:unambiguous-concat}
\begin{proof}
Let $\big( L_1, \CA{L}_1, \CA{L}_1 \cdot ( \prefix{\CA{L}_2} \setminus \{\epsilon\}) \big)$
be a partition of $\prefix{\CA{L}_1 \cdot \CA{L}_2}$ and
$L$ be a proper subset of $\CA{L}_1$.
Assume that there exists $w \in L$ and
$w' \in \CA{L}_1 \setminus L$ such that
$w=w'_{0, i}$, for some $i \in \{0, \ldots, \card{w'}-1\}$.
It follows that $w \in L_1$, because $w\neq w'$.
However, this contradicts the fact that $L_1$ and
$\CA{L}_1$ are disjoint.

Conversely, let $\CA{L}_1$ be unambiguous and
consider a word $w \in \prefix{\CA{L}_1 \cdot \CA{L}_2}$.
Assume that $w \in L_1 \cap \CA{L}_1$.
It follows that $\{w\}$ is a prefix language
for $\CA{L}_1\setminus \{w\}$, which contradicts with the
hypothesis that $\CA{L}_1$ is unambiguous.
Similarly, if we assume that there exists
$w \in \prefix{\CA{L}_1} \cap \big(\CA{L}_1\cdot (\prefix{\CA{L}_2} \setminus \{\epsilon\}) \big)$,
then there exists $w', w'' \in \CA{L}_1$ such that
$w'$ is a prefix of $w$, $w$ is a prefix of $w''$,
and $\card{w'} < \card{w} \leq \card{w''}$.
Thus, we arrive again at a contradiction with
the unambiguity of $\CA{L}_1$. Thus, the three
sets form a partition of $\prefix{\CA{L}_1 \cdot \CA{L}_2}$.
\end{proof}

\subsection{Proof of Prop.~\ref{th:tr-partial-order}}
\label{app:tr-partial-order}
\begin{proof}
The proof follows by structural induction over $AST(\phi)$.
The base case is trivial, since the leafs correspond to the {\em hold} operators.
For the induction step, the result follows trivially if the
intermediate node is associated with a Boolean or concatenation
operator. The case of a {\em within} operator follows from
Eq.~\eqref{eq:twtl-prop-tw-expand} and~\eqref{eq:twtl-prop-within-implication}
in Prop.~\ref{th:twtl-props}, i.e.
$[\phi(\BS{\tau})]^{[a, b + \tau_1]} \Implies [\phi(\BS{\tau}')]^{[a, b + \tau_1]} \Implies [\phi(\BS{\tau}')]^{[a, b + \tau'_1]}$,
where $a < b \in \BB{Z}_{\geq 0}$ and $\BS{\tau} \leq \BS{\tau}' \in \BB{Z}^m$.
We assumed without loss of generality that the first component of
the temporal relaxation vectors is assigned to the root node.
\end{proof}

%\section*{Acknowledgments}

%% Use plainnat to work nicely with natbib. 

\bibliographystyle{plainnat}
\bibliography{references}

\end{document}